\documentclass[a4paper,UKenglish,cleveref, autoref, thm-restate]{lipics-v2021} 
\newtheorem*{theorem*}{Theorem}

\nolinenumbers

\EventShortTitle{Full Version}
\ArticleNo{39}

\title{On the Number of Quantifiers as a Complexity Measure} %TODO Please add

\author{Ronald Fagin}{IBM Research-Almaden, USA}{fagin@us.ibm.com}{https://orcid.org/
0000-0002-7374-0347}{}

\author{Jonathan Lenchner\footnote{Corresponding author}}{IBM T.J. Watson Research Center, USA}{lenchner@us.ibm.com}{https://orcid.org/0000-0002-9427-8470}{}

\author{Nikhil Vyas}{MIT EECS, USA}{nikhilv@mit.edu}{0000-0002-4055-7693}{}

\author{Ryan Williams}{MIT CSAIL and EECS, USA}{rrw@umit.edu}{0000-0003-2326-2233}{}

\authorrunning{R. Fagin et al.}
\Copyright{Ronald Fagin, Jonathan Lenchner, Nikhil Vyas and Ryan Williams} %TODO mandatory, please use full first names. LIPIcs license is "CC-BY";  http://creativecommons.org/licenses/by/3.0/

\relatedversiondetails[cite=Fagin21]{``Extended Version''}{http://TBD}  %Must be filled out with correct information once available!



\category{} %optional, e.g. invited paper

\relatedversion{} %optional, e.g. full version hosted on arXiv, HAL, or other respository/website
\acknowledgements{}%optional

\usepackage{soul}

\newtheorem{prop}[lemma]{Proposition}

\newcommand{\AC}{\ensuremath{\mathsf{AC}}}
\newcommand{\NC}{\ensuremath{\mathsf{NC}}}
\newcommand{\FO}{\ensuremath{\mathsf{FO}}}

\newcommand{\bi}{\begin{itemize}}
\newcommand{\ei}{\end  {itemize}}
\newcommand{\bt}{\begin{tabbing}}
\newcommand{\et}{\end  {tabbing}}
\newcommand{\be}{\begin{enumerate}}
\newcommand{\ee}{\end  {enumerate}}

\newcommand{\fraisse}{Fra\"{i}ss\'{e} }

\newcommand{\set}[1]{\{#1\}}
\newcommand{\qr}{\mathit{qr}}

\newcommand{\ef}{Ehrenfeucht-\fraisse}
\newcommand{\edashf}{E-F }
\newcommand{\plog}{\text{polylog}}
\newcommand{\vs}{vs. }

\usepackage{mathtools}

\DeclarePairedDelimiter\ceil{\lceil}{\rceil}

\begin{document}

\maketitle

%\begin{abstract}
%In 1981, Immerman~\cite{Immerman81} introduced a game he called the ``Separability game'', somewhat analogous to \ef games. Unlike \ef games, which characterize the quantifier rank needed to describe a property $P$ in first-order logic~\cite{Ehr61, Fra54}, the Separability game characterizes the total number of quantifiers needed to describe such a property. These games lay unstudied until they were rediscovered in the recent paper of Fagin, Lenchner, Regan and Vyas~\cite{Fagin21}. In that paper, the authors reprove the fundamental equivalence theorem 
%connecting optimal game play with the minimum number of quantifiers needed to express a property of structures in first-order logic, and characterize the minimum number of quantifiers needed to distinguish linear orders of any two different sizes. 
%In this paper, we prove that there are properties of finite structures (in particular, finite graphs) for which the number of quantifiers needed to express the properties in first-order logic can be at least exponentially greater than the quantifier rank.
%We then extend the results in \cite{Fagin21} on linear orders in several directions. 
%First, we characterize the number of quantifiers needed to distinguish rooted trees %and forests 
%of two different depths. We then give lower and upper bounds, which are within a small constant factor of one another, for the number of quantifiers needed to express $s$-$t$ connectivity. 
%\end{abstract}
\begin{abstract}
In 1981, Neil Immerman described a two-player game, which he called the ``separability game'' \cite{Immerman81}, that captures the number of quantifiers needed to describe a property in first-order logic. Immerman's paper laid the groundwork for studying the number of quantifiers needed to express properties in first-order logic, but the game seemed to be too complicated to study, and the arguments of the paper almost exclusively used quantifier rank as a lower bound on the total number of quantifiers. However, last year Fagin, Lenchner, Regan and Vyas \cite{Fagin21} rediscovered the game, provided some tools for analyzing them, and showed how to utilize them to characterize the number of quantifiers needed to express linear orders of different sizes.  In this paper, we push forward in the study of number of quantifiers as a bona fide complexity measure by establishing several new results. First we carefully distinguish minimum number of quantifiers from the more usual descriptive complexity measures, minimum quantifier rank and minimum number of variables. Then, for each positive integer $k$, we give an explicit example of a property of finite structures (in particular, of finite graphs) that can be expressed with a sentence of quantifier rank $k$, but where the same property needs 
$2^{\Omega (k^2)}$ quantifiers to be expressed.
%$f(k)$ quantifiers to be expressed, where 
%$%f(k) = 2^{\Omega (k^2)}$.
%$f(k)$ is more than  exponential in $k$ (it is %$2^{\Omega (k^2)}$).
We next give the precise number of quantifiers needed to distinguish two rooted trees of different depths. Finally, we give a new upper bound on the number of quantifiers needed to express $s$-$t$ connectivity, improving the previous known bound by a constant factor.
\end{abstract}

%\hl{1. write the improved lower bound 2. prior work section 3. think about open problems.}

\newpage

% !TEX root = new_rev.tex

\section{Introduction}
In 1981 Neil Immerman described a two-player combinatorial game, which he called the ``separability game'' \cite{Immerman81}, that captures the number of quantifiers needed to describe a property in first-order logic (henceforth FOL).  In that paper Immerman remarked,

\begin{small}
\begin{quote}
``Little is known about how to play the separability game. We leave it here as a 
jumping off point for further research. We urge others to study it, hoping that the separability game may become a viable tool for ascertaining some of the lower bounds which are `well believed' but have so far escaped proof.''
\end{quote}
\end{small}

Immerman's paper laid the groundwork for studying the number of quantifiers needed to express properties in FOL, but alas, the game seemed too complicated to study and
%, with the exception of one result\footnote{The one result is the following: Let $Quants^S[f(n)]$ denote the family of properties expressible with $f(n)$ quantifiers in a language that includes a successor relation, $S(,)$ on pairs of elements of the structure. Then
%\begin{equation*}
%    NSPACE[f(n)] \subseteq Quants^S\Big[\frac{(f(n))^2}{\log{n}}\Big] %subseteq DSPACE[(f(n))^2].
%\end{equation*}} (which did not make use of games) 
the paper used the surrogate measure of quantifier rank, which provides a lower bound on the number of quantifiers, to make 
%virtually all of 
its arguments.
One of the reasons for the difficulty of directly analyzing the number of quantifiers is that the separability game is played on a pair $(\mathcal{A}, \mathcal{B})$ of \textit{sets} of structures, rather than on a pair of structures as in a conventional Ehrenfeucht-Fra\"{i}ss\'{e} game. However, last year Fagin, Lenchner, Regan and Vyas \cite{Fagin21} %independently 
rediscovered the games, provided some tools for analyzing them, and showed how to utilize them to characterize the number of quantifiers needed to express linear orders of different sizes.  In this paper, we push forward in the study of number of quantifiers as a bona fide complexity measure by establishing several new results, using these rediscovered games as an important, though not exclusive, tool. Although Immerman called his game the ``separability game,'' we keep to the more evocative ``multi-structural game,'' as coined in \cite{Fagin21}.

Given a property $P$ definable in FOL, let $Quants(P)$ denote the minimum number of quantifiers over all FO sentences that express $P$. \textbf{This paper exclusively considers expressibility in FOL.}
%we are generically interested in the expression that distinguishes $P$ from other structures not satisfying $P$, which utilizes a minimum number of quantifiers. We shall denote this quantity by $Quants(P)$.
$Quants(P)$ is related to two more widely studied descriptive complexity measures, the minimum quantifier rank needed to express $P$, and the minimum number of variables needed to express $P$. 
The quantifier rank of an FO sentence $\sigma$ is typically denoted by $qr(\sigma)$. We shall denote the \emph{minimum} quantifier rank over \emph{all} FO sentences describing the property $P$ by $Rank(P)$, and denote the minimum number of variables needed to describe $P$ by $Vars(P)$. When referring to a specific sentence $\sigma$, we shall denote the analogs of $Quants()$ and $Vars()$ by $quants(\sigma)$ and $vars(\sigma)$. (That is, $quants(\sigma), vars(\sigma)$ and $qr(\sigma)$ refer to the number of quantifiers, variables and quantifier rank of the particular sentence $\sigma$.) On the other hand, $Quants(P), Vars(P)$ and $Rank(P)$ refer to the minimum values of these quantities among all expressions describing $P$. Possibly there is one sentence establishing $Quants(P)$, another establishing $Vars(P)$, and a third establishing $Rank(P)$.
%
%How do these measures relate? 
We investigate the \emph{extremal} behavior of $Quants(P)$, via studying concrete properties $P$ for which $Quants(P)$ behaves differently from the other measures.

%First of all, since variables can be reused, note that $Quants(P) \geq Vars(P)$ for every property $P$, and $Quants(P)$ can be much larger than $Vars(P)$. 
%% ryan says: Maybe we can just state the above, and let the reader think about an example...
%Any sentence can be rewritten in prenex normal form, preserving the number of quantifiers, but potentially adding additional variables. For example, the sentence $\sigma = \exists x(\exists y E(x,y) \land \exists y E(y,x))$, of $3$ quantifiers and $2$ variables, which distinguishes the graph $B$ from the graph $A$ in Figure \ref{fig:lo3_vs_lo2}, can be rewritten in prenex normal form using $3$ quantifiers and $3$ variables, as $\sigma' = \exists x\exists y \exists  z(E(x,y) \land E(z,x))$. 
%Besides the obvious fact that $Vars(P) \leq Quants(P)$, 
%The following simple proposition observes that the minimum number of variables can be upper-bounded by minimum rank as well. % In fact, the following simple proposition shows  we have something a bit stronger than the obvious relation , as the following simple proposition shows.
First of all, for every property $P$, since every variable in a sentence describing $P$ is bound to a quantifier, and quantifiers can only be bound to a single variable, it must be that $Vars(P) \leq Quants(P)$. The following simple proposition observes that $Vars(P)$ is also upper bounded by $Rank(P)$.

\begin{prop}\label{vars_vs_rank_thm} For every property $P$: $Vars(P) \leq Rank(P)$.
\end{prop}

\begin{proof} We prove this result by showing that every \textit{formula} $\phi$, possibly with free variables, can be rewritten simply by changing the names of some of the variables, so that the number of \textit{bound} variables does not exceed the quantifier rank. Denote the minimum possible number of bound variables needed to express $\phi$ by $bdvars(\phi)$. %We proceed by induction on the length of $\phi$. 
If $\phi$ is a term, then $bdvars(\phi) = 0$ so $bdvars(\phi) = qr(\phi)$. Inductively, if $\phi$ is a formula satisfying $bdvars(\phi) \leq qr(\phi)$ and $\psi = \neg \phi$ then we also have $bdvars(\psi) \leq qr(\psi)$. Further, if $\phi$ satisfies $bdvars(\phi) \leq \qr(\phi)$ and $\psi$ satisfies $bdvars(\psi) \leq \qr(\psi)$, then consider $\phi \circ \psi$, for $\circ \in \{\vee, \land\}$.  Let $m = \min(bdvars(\phi), bdvars(\psi))$ and without loss of generality assume $bdvars(\phi) \leq bdvars(\psi)$. Then, if $x_1,\ldots,x_m$ are the names of the bound variables in $\phi$, we can use these same variable names for the first $m$ bound variables in $\psi$, and as well in $\phi \circ \psi$, with no change in meaning. Then $bdvars(\phi \circ \psi) = bdvars(\psi) \leq qr(\psi) \leq qr(\phi \circ \psi)$. Lastly, suppose $bdvars(\phi) \leq qr(\phi)$ and we add a quantifier $Q \in \{\forall, \exists\}$ over a free variables $x$ in $\phi$ to form $\psi = Qx\phi$. The variable name $x$ may or may not be distinct from any of the previously bound variable names in $\phi$ so that $bdvars(\psi) \leq bdvars(\phi) + 1$, while $qr(\psi) = qr(\phi) + 1$. Thus again $bdvars(\psi) \leq qr(\psi)$. Since for a sentence $\sigma$, we have $vars(\sigma) = bdvars(\sigma)$, the lemma is established.
\end{proof}

As a corollary, since clearly $Rank(P) \leq Quants(P)$, we have: 
\smallskip
\begin{align}\label{ineq}
Vars(P) \leq Rank(P) \leq Quants(P).
\end{align}
Furthermore, it follows from Immerman~\cite[Prop. 6.15]{Imm99} that $Quants(P)$ and $Rank(P)$ can both be arbitrarily larger than $Vars(P)$. When the property $P$ is $s$-$t$ connectivity up to path length $k$, Immerman shows that $Vars(P) \leq 3$, yet $Rank(P) \geq \log_2(k)$. %can be expressed with a sentence using only $3$ variables and quantifier rank $\log_2(k)$, and $\log_2(k)$ is also the minimum quantifier rank needed to express this property. % regardless of the number of variables. 

\paragraph*{Summary of Results} From equation \eqref{ineq}, we see that the number of quantifiers needed to express a property is lower-bounded by the minimum quantifier rank and number of variables. How much larger can $Quants(P)$ be, compared to the other two measures? It is known (see \cite{dawgrokre+07}) that there \emph{exists} a fixed vocabulary $V$ and an infinite sequence $P_1, P_2,...$ of properties such that $P_k$ is a property of finite structures with vocabulary $V$ such that $Rank(P_k) \leq k$, yet $Quants(P_k)$ is not an elementary function of $k$. However, the existence of such $P_k$ are proved via counting arguments.
We provide an \emph{explicitly computable} sequence of properties $\{P_k\}$ with a high growth rate in terms of the number of quantifiers required. (By ``explicitly computable'', we mean that there is an algorithm $A$ such that, given a positive integer $k$, the algorithm $A$ prints a FO sentence $\sigma_k$ with quantifier rank $k$ defining the property $P_k$, in time polynomial in the length of $\sigma_k$.)

%\begin{theorem}[See  Theorem~\ref{thm:discrepency}, Section~\ref{sec:discrepency}]
%\label{thm:intro1}
%There is an explicitly computable sequence of properties $\{P_k\}$ such that for all $k$ we have $Rank(P_k) \leq k$, yet %$Quants(P_k) \geq 2^{\Omega(k^2)}$.
%\end{theorem}
\begin{theorem*}[Theorem~\ref{thm:discrepency}, Section~\ref{sec:discrepency}]
\label{thm:intro1}
There is an explicitly computable sequence of properties $\{P_k\}$ such that for all $k$ we have $Rank(P_k) \leq k$, yet $Quants(P_k) \geq 2^{\Omega(k^2)}$.
\end{theorem*}

%In the proof of Theorem \ref{thm:discrepency} in Section \ref{sec:discrepency}, we give an example of a property $P$ for which $Vars(P) = k$ and $Rank(P) = k$, but for which $Quants(P) = 2^{\Omega(k^2)}$. 
%It remains to find some upper bound on $\#q$ in terms of $qr$ and $\#v$, possibly incorporating something about the number and arities of the relations in a chosen vocabulary. 

%\subsection{Paper Overview}

Next, we give an example of a setting in which one can completely nail down the number of quantifiers that are necessary and sufficient for expressing a property. Building on Fagin \emph{et al.}~\cite{Fagin21}, which gives results on the number of quantifiers needed to distinguish linear orders of different sizes, we study the number of quantifiers needed to distinguish rooted trees of different depths.

Let $t(r)$ be the maximum $d$ such that there is a formula with $r$ quantifiers that can distinguish rooted trees of depth $d$ (or larger) from rooted trees of depth less than $d$. Reasoning about the relevant multi-structural games, we can completely characterize $t(r)$, as follows.

\begin{theorem*}[Theorem~\ref{thm:t_explicit}, Section~\ref{section:rooted_trees_and_forests}]
\label{thm:intro2}
For all $r \geq 1$ we have 
%\begin{align*}
%    t(2r) = \frac{7\cdot4^r}{18} + \frac{4r}{3} - \frac{8}{9}, \\
%    t(2r+1) = \frac{8\cdot4^r}{9} + \frac{4r}{3} - \frac{8}{9}.
%\end{align*}
\begin{equation*}
    t(2r) = \frac{7\cdot4^r}{18} + \frac{4r}{3} - \frac{8}{9},~~~~~~~ 
    t(2r+1) = \frac{8\cdot4^r}{9} + \frac{4r}{3} - \frac{8}{9}.
\end{equation*}
\end{theorem*}
It follows from the above theorem that we can distinguish (rooted) trees of depth at most $d$ from trees of depth greater than $d$ using only $\Theta(\log d)$ quantifiers, and we can in fact pin down the \emph{exact} depth that can be distinguished with $r$ quantifiers. This illustrates the power of multi-structural games, and gives hope that more complex problems may admit an exact number-of-quantifiers characterization.

Next, we consider the question of how many quantifiers are needed to express that two nodes $s$ and $t$ are connected by a path of length at most $n$, in directed (or undirected) graphs. In our notation, we wish to determine $Quants(P)$ where $P$ is the property of $s$-$t$ connectivity via a path of length at most $n$.
Considering the significance of $s$-$t$ connectivity in both descriptive complexity and computational complexity, we believe this is a basic question that deserves a clean answer. It follows from the work of Stockmeyer and Meyer %on $\mathsf{PSPACE}$-completeness~\cite[Theorem 4.3]{StockmeyerM73}, 
that $s$-$t$ connectivity up to path length $n$ can be expressed with $3\log_2(n) + O(1)$ quantifiers. As mentioned earlier, $s$-$t$ connectivity is well-known to require quantifier rank at least $\log_2(n) - O(1)$. We manage to reduce the number of quantifiers necessary for $s$-$t$ connectivity.

\begin{theorem*}[Theorem~\ref{thm:st-con-upper}, Section~\ref{sec:st-conn}] \label{thm:intro3} The number of quantifiers needed to express $s$-$t$ connectivity is at most $3\log_3(n) + O(1) \approx 1.893\log_2(n) + O(1)$.
\end{theorem*}

The remainder of this manuscript proceeds as follows. In the next subsection we describe multi-structural games and compare them to %the more well-known 
Ehrenfeucht-Fra\"{i}ss\'{e} games.  In the subsection that follows we review related work in complexity. We then prove the theorems mentioned above. 
In Section~\ref{sec:discrepency} we prove Theorem~\ref{thm:discrepency}. %give an explicit example that shows that for each positive integer $k$, there is a property of finite structures (in particular, of finite graphs) that can be expressed with a sentence of quantifier rank $k$, but where the same property needs $2^{\Omega (k^2)}$ quantifiers to be expressed.
%kf(k)$ quantifiers to be expressed, where 
%$f(k) = 2^{\Omega (k^2)}$.
In Section~\ref{section:rooted_trees_and_forests} we prove Theorem~\ref{thm:t_explicit}. % give the precise number of quantifiers needed to distinguish two rooted trees of different depths (). 
In Section~\ref{sec:st-conn} we prove Theorem~\ref{thm:st-con-upper}.
%give upper and lower bounds on the number of quantifiers needed to express $s$-$t$ connectivity.
In Section~\ref{sec:fina- comments}, we 
%conclude with some comments and suggestions for future %research.
give final comments and suggestions for future research.

\subsection{Multi-Structural Games}

The standard  \ef game (henceforth \edashf game) is played by ``Spoiler'' and ``Duplicator'' on a pair $(A,B)$ of structures over the same FO vocabulary $V$, for a specified number $r$ of \emph{rounds}.  
If $V$ contains constant symbols $\lambda_1,...,\lambda_k$, then designated (``constant'') elements $c_i$ of $A$, and $c'_i$ of $B$, must be associated with each $\lambda_i$.
In each round, Spoiler chooses an element from $A$ or from $B$, and Duplicator replies by choosing an element from the other structure.  In this way, they determine sequences of elements $a_1,\dots,a_r, c_1,\dots,c_k$ of $A$ and $b_1,\dots,b_r, c'_1,\dots,c'_k$ of $B$, which in turn define substructures $A'$ of $A$ and $B'$ of $B$.  Duplicator wins 
%the play of the game 
if the function given by $f(a_i) = b_i$ for $i = 1,\dots,r,$ and $f(c_j) = c'_j$ for $j = 1,\dots,k$, is an isomorphism of $A'$ and $B'$.  Otherwise, Spoiler wins.  

The equivalence theorem for \edashf games \cite{Ehr61,Fra54} characterizes the minimum \emph{quantifier rank} of a sentence $\phi$ over $V$ that is true for  $A$ but false for $B$.  The quantifier rank $\qr(\phi)$ is defined as zero for a quantifier-free sentence $\phi$, and inductively:
\begin{eqnarray*}
\qr(\neg\phi) &=& \qr(\phi),\\
\qr(\phi \vee \psi) &=& \qr(\phi \land \psi) ~=~ \max\set{\qr(\phi),\qr(\psi)},\\
\qr(\forall x\phi) &=& \qr(\exists x\phi) ~=~ \qr(\phi) + 1.
\end{eqnarray*}
%Some authors use the phrase \textit{quantifier depth} rather than quantifier rank.

%\bigskip
\begin{theorem}[\cite{Ehr61,Fra54}]  \textbf{Equivalence Theorem for E-F Games:}\label{thm:ef}
Spoiler wins the $r$-round \edashf game on $(A,B)$ if and only if there is a sentence $\phi$ of quantifier rank at most $r$ such that $A \models \phi$ while $B \models \neg\phi$.
\end{theorem}

In this paper we make use of a variant of \edashf games, which have come to be called \emph{multi-structural games} \cite{Fagin21}. Multi-structural games (henceforth M-S games) make Duplicator more powerful and can be used to characterize the \emph{number} of quantifiers, rather than the quantifier \emph{rank}.
%It is straightforward to see that the minimum number of quantifiers needed to define a property $P$ is the same as the minimum size of the quantifier prefix of a sentence in prenex normal form that is needed to define property $P$. This is because converting a sentence into prenex normal form does not increase the number of quantifiers.
In an M-S game there are again two players, Spoiler and Duplicator, and there is a fixed number $r$ of rounds.  Instead of being played on a pair %$\mathsf{A}$, $\mathsf{B}$
$(A,B)$
of structures with the same vocabulary (as in an \edashf game), the M-S game is played on a pair $({\cal A}, {\cal B})$ of \textit{sets of structures}, all with the same vocabulary.  
For $k$ with $0 \leq k \leq r$, by a %{\em labeled$_k$}  
\textit{labeled structure} after $k$ rounds, we mean a structure along with a labeling of which elements were selected from it in each of the first $k$ rounds. 
Let ${\cal{A}}_0 = \cal{A}$ and ${\cal{B}}_0 = \cal{B}$.
Thus, ${\cal{A}}_0$ represents the labeled structures from $\cal A$ after 0 rounds, and similarly for ${\cal{B}}_0$ -- in other words nothing is yet labelled except for constants.
If $1 \leq k <r$, let ${\cal{A}}_k$ be the labeled structures originating from $\cal A$ after $k$ rounds, and similarly for ${\cal{B}}_k$.
In round $k+1$, 
Spoiler either chooses an element from each member of ${\cal{A}}_k$, thereby creating ${\cal{A}}_{k+1}$,
or chooses an element from each member of ${\cal{B}}_k$, thereby creating ${\cal{B}}_{k+1}$.
Duplicator responds as follows. 
Suppose that Spoiler chose an element from each member of ${\cal{A}}_k$, thereby creating ${\cal{A}}_{k+1}$.   Duplicator can then make multiple copies of each labeled structure of ${\cal{B}}_k$, and choose an element from each copy,
%of a member of ${\cal{B}}_k$, 
thereby creating ${\cal{B}}_{k+1}$.
Similarly, if Spoiler chose an element from each member of ${\cal{B}}_k$, thereby creating ${\cal{B}}_{k+1}$, Duplicator can then make multiple copies of each labeled structure of ${\cal{A}}_k$, and choose an element from each copy,
%of a member of ${\cal{A}}_k$, 
thereby creating ${\cal{A}}_{k+1}$.
Duplicator wins if there is some labeled structure $A$ in ${\cal{A}}_{r}$ and some labeled structure $B$ in ${\cal{B}}_{r}$ where the labelings give a partial isomorphism.  
%That is, if $a_i$ is the point selected in $A$ in round $i$, for $1 \leq i \leq r$, and if  $b_i$ is the point selected in $B$ in round $i$, for $1 \leq i \leq r$, and if we let $A'$ be the substructure of $A$ (ignoring the labelings) generated by $\set{a_1, \ldots, a_r}$ and $B'$ the substructure of $B$ (ignoring the labelings) generated by $\set{b_1, \ldots, b_r}$, then the function $f$ where $f(a_i) = b_i$ for $i = 1,\dots,r$ is an isomorphism of $A'$ and $B'$. 
Otherwise, Spoiler wins.  

In discussing M-S games we sometimes think of the play of the game by a given player, in a given round, as taking place on one of two ``sides'', the $\cal{A}$ side or the $\cal{B}$ side, corresponding to where the given player plays from on that round.

Note that on each of Duplicator's moves, Duplicator can make ``every possible choice," via the multiple copies. 
Making every possible choice creates what we call the \emph{oblivious strategy}. Indeed, Duplicator has a winning strategy if and only if the oblivious strategy is a winning strategy.
%for Duplicator. 
  
%   $of $\cal B$ (``labeled" meanikng that it continues to reflect the moves made on it so far), and then chooses 
%  an element from each member $B$ of the newly enlarged set $\cal B$.
%  Analogously, if Spoiler chose an element from each member $B$ of $\cal B$, then 
%  Duplicator can then make multiple copies of each labeled member $A$ of $\cal A$, and then chooses 
%  an element from each member $A$ of the newly enlarged set $\cal A$.
%   In this way, they determine sequences $a_1,\dots,a_r$ of elements from each member $A$ in the possibly enlarged set $\cal A$, and $b_1,\dots,b_r$ from each member of the possibly enlarged set $\cal B$.
%   As before, repetitions are allowed, and these choices define substructures $A_r$ of each
%   $A$ of the possibly enlarged $\cal A$ and each $B$ of the possibly enlarged 
%   $\mathsf{A}$ and $\mathsf{B_r}$ of $\mathsf{B}$.  Duplicator wins the play of the game if the function $f(a_i) = b_i$ for $i = 1,\dots,r$ is an isomorphism of $\mathsf{A_r}$ and $\mathsf{B_r}$.  Else, Spoiler wins.  

The following equivalence theorem, proved in \cite{Immerman81, Fagin21}, is the analog of Theorem~\ref{thm:ef} for \edashf games.
%\hl{below theorem should say quantifier depth?}
\begin{theorem}[\cite{Immerman81, Fagin21}] \textbf{Equivalence Theorem for Multi-Structural Games:} \label{thm:main1}
Spoiler wins the $r$-round M-S game on
$(\cal{A}, \cal{B})$
if and only if there is a sentence $\phi$ with at most $r$ quantifiers such that $A \models \phi$ for every $A \in {\cal A}$ while $B \models \neg\phi$ for every $B \in {\cal B}$.
\end{theorem}

In \cite{Fagin21} the authors provide a simple example of a property $P$ of a directed graph that requires $3$ quantifiers but which can be expressed with a sentence of quantifier rank $2$. $P$ is the property of having a vertex with both an in-edge and an out-edge. $P$ can be expressed via the sentence
$\sigma = \exists x(\exists y E(x,y) \land \exists y E(y,x))$,    
where $E(,)$ denotes the directed edge relation.  In \cite{Fagin21} it is shown that while Spoiler wins a 2-round E-F game on the two graphs $A$ and $B$ in Figure \ref{fig:lo3_vs_lo2},
\begin{figure} [h]
\centerline{\scalebox{0.45}{\includegraphics{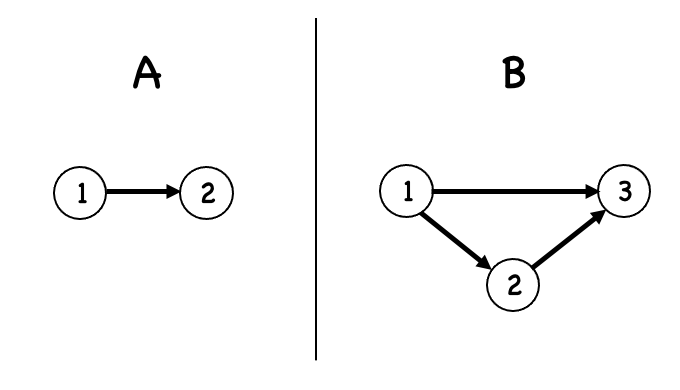}}}
\caption{The graph $B$, on the right, contains a vertex with both an in-edge and an out-edge, while the graph $A$,  on the left, does not.}
\label{fig:lo3_vs_lo2}
\end{figure}
 Duplicator wins the analogous 2-round M-S game starting with these two graphs. Hence, by Theorem \ref{thm:main1}, the property $P$ is not expressible with just $2$ quantifiers.

%The following is an immediate consequence of Theorem \ref{thm:main1}:
%
%\begin{cor} \label{ms-key-cor}
%There is a sentence with $r$ quantifiers that distinguishes structures satisfying a property $P$ from those not satisfying $P$ iff for all sets $(\mathcal{A}, \mathcal{B})$ of structures such that $A \models P$ and $B \models \neg P$ for all $A \in \mathcal{A}$ and $B \in \mathcal{B}$ Spoiler wins the associated $r$-round MS game on these sets of structures.
%\end{cor}

%\textbf{Prior Work}

%\bigskip

%In this paper we investigate the total number of quantifiers as a complexity measure in further depth, first studying rooted trees of different depths and then the classical $s$-$t$ connectivity problem from this vantage point. We will frequently employ multi-structural games to aid us in this study. 
%
%\hl{should we state the results of the paper formally here?}

\subsection{Related Work in Complexity}

Trees are a much studied data structure in complexity theory and logic. It is well known that it is impossible, in FOL, to express that a graph with no further relations is a tree \cite[Proposition 3.20]{Lib12}. We note, however, that given a partial ordering on the nodes of a graph, it is easy to express in FOL the property that the partial ordering gives rise to a tree. The relevant sentence expresses that there is a root (i.e., a greatest element) from which all other nodes descend, and if a node $x$ has nodes $y$ and $z$ as  distinct ancestors then one of $y$ and $z$ must have the other as its own ancestor. Hence the needed sentence is the conjunction of the following two sentences: 
\begin{flalign*}
    &\exists x \forall y(y \neq x \rightarrow y < x), \\
    %&\forall x \forall y ((x<y) \rightarrow \neg(y<x)), \\
    &\forall x \forall y \forall z((x < y \land x < z \land y \neq z) \rightarrow (y < z \vee z < y)).
\end{flalign*}
There are also interesting models of computation and logics based on trees.  See, for example, the literature on Finite Tree Automata \cite{Comon99} and Computational Tree Logic \cite{Clarke81}.
%Computational Tree Logic, for example, is a well developed research area in its own right \cite{Clarke81}.

%\hl{Add the trees are well studied in logic? next thing to do after linear orders which were studied in last paper}

We now discuss $s$-$t$ connectivity. In this paragraph only, $n$ denotes the number of nodes in the graph and $k$ the number of edges in a shortest path from $s$ to $t$. The $s$-$t$ connectivity problem has been studied extensively in both logic \cite{AjtaiFagin90,Imm99} and complexity theory. Most complexity studies of this problem have focused on space and time complexity. Directed $s$-$t$ connectivity is known to be $\mathsf{NL}$-complete~(see for example Theorem 16.2 in \cite{PAPA:1995}), while undirected $s$-$t$ connectivity is known to be in $\mathsf{L}$~\cite{REINGOLD:2005}. Savitch~\cite{Savitch70} proved that $s$-$t$ connectivity can be solved in $O(\log^2(n))$ space and $n^{\log_2(n)(1+O(1))}$ time. Recent work of Kush and Rossman~\cite{KushR20} has shown that the \emph{randomized} $\AC^0$ formula complexity of $s$-$t$ connectivity is at most size $n^{0.49 \log_2(k) + O(1)}$, a slight improvement. Barnes, Buss, Ruzzo and Schieber~\cite{BarnesBRS98} gave an algorithm running in both sublinear space and polynomial time for $s$-$t$ connectivity. Gopalan, Lipton, and Meka~\cite{GopalanLM03} presented randomized algorithms for solving $s$-$t$ connectivity with non-trivial time-space tradeoffs.
The $s$-$t$ connectivity problem has also been studied from the perspective of circuit and formula depth. For the weaker model of $\mathsf{AC}^0$ formulas an $n^{\Omega(\log(k))}$ size lower bound is known to hold unconditionally \cite{BeameIP98, ChenOST16, Rossman14}.

%for the $\leq k$ connectivity problem.\footnote{Here we are promised that if there is a path from $s$ to $t$ then there is a path of length $\leq k$ from $s$ to $t$}.

%In logic, undirected $s$-$t$ connectivity is known to be inexpressible in first-order Logic~\cite{Lib12} but expressible in Monadic Second Order Logic (MSO), while directed $s$-$t$ connectivity is also inexpressible in Monadic 2nd Order Logic (MSO)~\cite{AJTAI:1990}.

There is also a natural and well-known correspondence with the number of quantifiers in FOL and circuit complexity, in particular with the circuit class $\AC^0$ (constant-depth circuits comprised of NOT gates along with unbounded fan-in OR and AND gates). For example, Barrington, Immerman, and Straubing~\cite{BarringtonIS88} proved that $\text{uniform}_{\FO}\text{-}\AC^0 = \FO[<, \text{BIT}]$, thus characterizing the problems solvable in uniform $\AC^0$ by those expressible in FOL with ordering and a BIT relation. 

More generally it is known that $\text{uniform}_{\FO}\text{-}\AC[t(n)] = \FO[<, \text{BIT}][t(n)]$~(\cite{Imm99}, Theorem 5.22), i.e., FO formulas over ordering and BIT relations, defined via constant-sized blocks that are ``iterated'' for $O(t(n))$ times, are equivalent in expressibility with $\AC$ circuits of depth $O(t(n))$. (See Appendix~\ref{app:eq} for a more detailed statement.) Generally speaking, the number of quantifiers of FOL sentences (with a regular form) roughly corresponds to the \emph{depth} of a (highly uniform) $\AC^0$ circuit deciding the truth or falsity of the given sentence. Thus the number of quantifiers can be seen as a %natural 
proxy for ``uniform circuit depth''.%\footnote{One might think of circuit depth as naturally corresponding to quantifier \emph{rank}; indeed, the authors have thought this was the natural correspondence. But in the \emph{uniform} cases described above, circuit depth actually corresponds to the number of quantifiers.}

%\hl{number of quantifiers corresponds to a particularly structured subset of AC0 - didn't understand this}
%\hl{added section below}

%\section{On the Discrepancy Between Quantifier Rank and Number of Quantifiers} \label{sec:discrepency}
%\section{On the Difference in Magnitude Between Quantifier Rank and Number of Quantifiers} \label{sec:discrepency}
\section{Difference in Magnitude: Quantifier Rank vs.\ Number of Quantifiers} \label{sec:discrepency}
Let $V$ be a vocabulary with at least one relation symbol  with arity at least 2.
It is known \cite{dawgrokre+07} that the number of inequivalent sentences in vocabulary $V$ with quantifier rank $k$ is not an elementary function of $k$ (that is, grows faster than any tower of exponents).
Since the number of sentences in vocabulary $V$ with $k$ quantifiers is at most only double exponential in $k$ (e.g., a function that grows like $2^{2^{p(k)}}$ for some polynomial $p(k)$ -- see Appendix \ref{app:double_exp_upper_bd} for a proof), 
%end of footnote
it follows by a counting argument that for each positive integer $k$, there is a property $P$ of finite structures with vocabulary $V$ that can be expressed by a sentence of quantifier rank $k$, but where
the number of quantifiers needed to express $P$ is not an elementary function of $k$. However, to our knowledge, up to now no \emph{explicit} examples have been given of a property $P$ where the quantifier rank of a sentence to express $P$ is $k$, but   where the number of quantifiers needed to express the property $P$ is at least exponential in $k$.
In the proof of the following theorem, we give such an explicit example.%\footnote{Probably the reason no explicit example was given before, was due to the lack of use of tools needed to count the required number of quantifiers.}

Let $f_V(k)$ be the number of structures with $k$ nodes up to isomorphism in vocabulary $V$  (such as the number of non-isomorphic graphs with $k$ nodes). Note that in the case of graphs (a single binary relation symbol), $f_V(k)$  is asymptotic to $(2^{k^2})/k!$ \cite{Harary58},
and Stirling’s formula implies that 
$f_V(k) = 2^{\Omega (k^2)}$).
We have the following theorem.

\begin{theorem} \label{thm:discrepency}
% Assume that the vocabulary $V$ contains at least one relation symbol with arity at least 2.
% For each positive integer $k$, there is a property $P$ of finite structures with vocabulary $V$ such that $P$ can be described in first-order logic by a sentence with quantifier rank $k$, while the minimum number of quantifiers needed to describe $P$ in first-order logic is $k f_V(k-1)$.
Assume that the vocabulary $V$ contains at least one relation symbol with arity at least 2. There is an algorithm such that given a positive integer $k$, the algorithm produces a FO sentence $\sigma$ of quantifier rank $k$ where the minimum number of quantifiers needed to express $\sigma$ in FOL is $k f_V(k-1)$, which grows like $2^{\Omega (k^2)}$, and where the algorithm runs in time polynomial in the length of $\sigma$.
\end{theorem} 

\begin{proof} 
For simplicity, let us assume that the vocabulary $V$ consists of a single binary relation symbol, so that we are dealing with graphs.  It is straightforward to modify the proof to deal with an arbitrary vocabulary with at least one relation symbol of arity at least 2. Let us write $f$ for $f_V$.
Let $C_1, \ldots, C_{f(k-1)}$ be the $f(k-1)$ distinct graphs up to isomorphism with $k-1$ nodes.
For each $j$ with $1 \leq j \leq f(k-1)$, derive the graph $D_j$ that is obtained from $C_j$ by adding one new node with a single edge to every node in $C_j$.
Thus, $D_j$ has $k$ nodes. 
$D_j$ uniquely determines $C_j$, since $C_j$ is obtained from $D_j$ by removing a node $a$ that has a single edge to every remaining node; even if there were two such nodes $a$, the result would be the same. 
Therefore, there are $f(k-1)$ distinct graphs $D_j$.
We now give our sentence $\sigma$.
%that defines the property $P$ in the statement of the %theorem. 
Let $\sigma_j$ be the sentence $\exists x_1 \cdot\cdot\cdot \exists x_k \tau_j(x_1,\ldots,x_k)$, which expresses that there is a graph with a subgraph isomorphic to $D_j$.
Then the sentence $\sigma$ is the conjunction of the sentences $\sigma_j$ for $1 \leq j \leq f(k-1)$.
Since the sentence $\sigma$ is of length $2^{\Omega(k^2)}$,
it is not hard to verify that this sentence can be generated by an algorithm running in polynomial time in the length of the sentence (there is enough time to do all of the isomorphism tests by a naive algorithm).

The sentence $\sigma$ has quantifier rank $k$.
As written, this sentence has $kf(k-1)$ quantifiers.  
%We shall show that this is the minimum number of quantifiers needed to express $\sigma$.
%We show that $kf(k-1)$ is the minimum number of quantifiers needed to express $\sigma$ using a multi-structural game argument in Appendix \ref{app:comp_proof_rons_thm}
Let $A$ be the disjoint union of $D_1, \ldots, D_{f(k-1)}$.
If $p$ is a point in $A$, define $B_p$ to be the result of deleting the point $p$ from $A$.
Let $\cal{A}$ consist only of $A$, and let 
$\cal{B}$ consist of the graphs $B_p$ for each $p$ in $A$.
If $p$ is in the connected component $D_j$ of $A$, then $B_p$ does not have a subgraph isomorphic to $D_j$.
Hence, no member of ${\cal B}$ satisfies $\sigma$.
Since the single member $A$ of ${\cal A}$ satisfies $\sigma$, and since no member of ${\cal B}$ satisfies $\sigma$, we can make use of M-S games played on $\cal{A}$ and $\cal{B}$ to find the number of quantifiers needed to express $\sigma$.%\footnote{For an explanation of why this argument would not have gone through had we replaced $D_1,...,D_{f(k-1)}$ with the distinct graphs $C'_1,...,C'_{f(k)}$ up to isomorphism on $k$ nodes, see Appendix \ref{app:proof_fine_point}.}
%\footnote{To see why we made use of the graphs $D_j$ rather than the graphs $C_j’$ where $C_1', \ldots, C_{f(k)}’$ are the distinct graphs up to isomorphism with $k$ nodes, let $\tau_j '(x_1, \ldots, x_k)$ be a first-order sentence that completely describes $C_j’$ up to isomorphism, for $1 \leq j \leq f(k)$. Let $\sigma’$ be the conjunction of the sentences $\exists x_1 \ldots \exists x_k \tau_j’ (x_1, \ldots, x_k)$. Let $A’$ be the disjoint union of the graphs $C_j’$, and let $B_p’$ be the result of deleting the point $p$ from $A’$. If $p$ is an isolated vertex of $C_j’$, then $B_p’$  still contains a subgraph $H$  where $H$ is isomorphic to $C_j'$: this graph $H$ has $k-1$ of its vertices from what is left of $C_j’$  after removing the isolated vertex $p$ from $C_j’$,  and the other vertex of $H$ is an isolated vertex in one of the other graphs $C_i’$.  So both $A'$ and $B_p’$ satisfy $\sigma’$, which would invalidate a multi-structural  game argument.}

Assume that we have labeled copies of $A$ and the various $B_p$’s after $i$ rounds of an M-S game played on $\cal{A}$ and $\cal{B}$.
The labelling tells us which points have been selected in each of the first $i$ rounds.
Let us say that a labeled copy of $A$ and a labeled copy of $B_p$ are {\em in harmony\/} after $i$ rounds if the following holds.
For each $m$ with $1 \leq m \leq i$,
if $a$ is the point labeled $m$ in $A$,
and $b$ is the point labeled $m$ in $B_p$, then $a=b$.
In particular, if the labeled copies of $A$ and $B_p$ are in harmony, then there is a partial isomorphism between the labeled copies of $A$ and $B_p$.

Let Duplicator have  the following strategy. 
Assume first that in round $i$,  Spoiler selects in $\cal{A}$, and selects a point $a$ from a labeled member $A$ of $\cal{A}$.
Then Duplicator (by making extra copies of labeled graphs in $\cal{B}$ as needed) does the following for each labeled $B_p$ in $\cal{B}$.  
If $a \neq p$, 
and if the labeled copies of $A$ and $B_p$ before round $i$ are in harmony, then Duplicator selects $a$ in $B_p$, which maintains the harmony.
If $a=p$, or
if the labeled $A$ and $B_p$ before round $i$ are not in harmony,
then Duplicator makes an arbitrary move in $B_p$.
%, and ignores this labeled $B_p$ from then on. 

Assume now that in round $i$,  Spoiler selects in $\cal{B}$. 
When Spoiler selects the point $b$ from
a labeled copy of $B_p$,
then for each labeled $A$ from $\cal{A}$, if the labeled copy of $A$ is in harmony with the labeled copy of $B_p$ before round $i$, then Duplicator selects $b$ in $A$, and thereby maintains the harmony.
We shall show shortly (Property * below) that in every round, each labeled member of $\cal{A}$ is in harmony with a labeled member of $\cal{B}$, so in the case we are now considering where Spoiler selects in $\cal{B}$, 
Duplicator does select a point in round $i$ in each labeled member of $\cal{A}$. 

We prove the following by induction on rounds:  
%
%\smallskip

\noindent \textbf{Property *}: If $A$ is a labeled graph in $\cal{A}$ and if point $p$ in $A$ was not selected in the first $i$ rounds, then there is a labeled copy of $B_p$ that is in harmony with $A$ after $i$ rounds. 
%
%\smallskip

Property * holds  after 0 rounds (with no points selected). Assume that Property * holds after $i$ rounds; we shall show that it holds after $i+1$ rounds.  There are two cases, depending on whether Spoiler moves in $\cal{A}$ or in $\cal{B}$ in round $i+1$.  Assume first that Spoiler moves in $\cal{A}$ in round $i+1$. Assume that point $p$ was not selected in $A$ after $i+1$ rounds. By inductive  assumption,
there are labeled versions of $A$ and $B_p$ that are in harmony after $i$ rounds.  So by Duplicator’s strategy, labeled versions of of $A$ and $B_p$ are in harmony after $i+1$ rounds. 
Now assume that Spoiler moves in $\cal{B}$  in round i+1. For each labeled graph $A$ in $\cal{A}$, if a labeled $B_p$ is in harmony with the labeled $A$ after $i$ rounds, then by Duplicator’s strategy, the harmony continues between the labeled $A$ and $B_p$  after $i+1$ rounds.  So Property * continues to hold after $i+1$ rounds.  This completes the proof of Property *.

After $(k f(k-1))-1$  rounds, pick an arbitrary labeled  graph $A$ in $\cal{A}$. 
Since at most $(k f(k-1))-1$  points have been selected after $(k f(k-1))-1$  rounds, and since $A$ contains $k f(k-1)$  points (because it is the disjoint union of $f(k-1)$ graphs each with $k$ points), there is some point $p$ that was not selected in $A$ in the first $(k f(k))-1$  rounds.  
Therefore, by Property *, a labeled version of $A$ and of $B_p$ are in harmony after 
$(k f(k))-1$  rounds, and hence there is a partial isomorphism between the labeled $A$ and $B_p$.
So Duplicator wins the 
$(k f(k-1))-1$  round M-S game!
Therefore, by Theorem~\ref{thm:main1},
the number of quantifiers needed to express $\sigma$
%the property $P$ 
is more than
$(k f(k-1))-1$.
Since $\sigma$ has $k f(k-1)$ quantifiers %(via the sentence $\sigma$), 
it follows that the minimum number of quantifiers need to express $\sigma$ is exactly $k f(k-1)$.
\end{proof}

\section{Rooted Trees}  \label{section:rooted_trees_and_forests}

Our aim in this section is to establish the minimum number of quantifiers needed to distinguish rooted trees of depth at least $k$ from those of depth less than $k$ using first-order formulas, given a partial ordering on the vertices induced by the structure of the rooted tree.
%Our aim in this section is to establish the number of quantifiers needed to distinguish rooted trees of depth $k$ and larger from those of depth smaller than $k$ using first-order formulas, given a partial ordering on the vertices induced by the structure of the rooted tree. 
Figure \ref{fig:basic_rooted_tree}
\begin{figure} [h]
\centerline{\scalebox{0.35}{\includegraphics{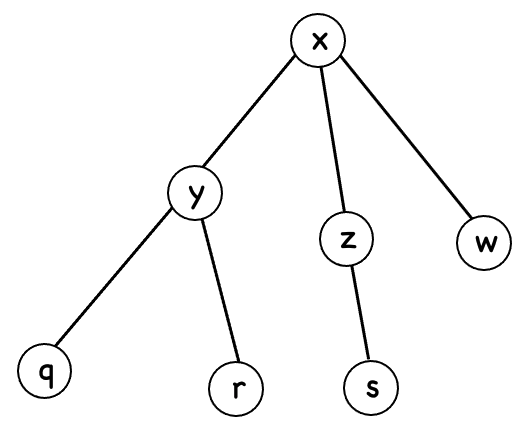}}}
\caption{A rooted tree with designated root node $x$ and depth $3$.}
\label{fig:basic_rooted_tree}
\end{figure}
gives an example of a tree where we designate $x$ as the root node. We define the depth of such a tree to be the maximum number of nodes in a path from the root to a leaf, where all segments in the path are directed from parent to child. Although it is more customary to denote the depth of a tree in terms of the number of \textit{edges} along such a path, we keep to the above definition because we will often run into the special case of linear orders, which we view as trees in the natural way, and linear orders are characterized by their size (number of nodes) and we would like the size of a liner order to correspond to the depth of the associated tree. Let us denote the tree rooted at $x$ by $T_x$. We make the arbitrary choice that the node $x$ is the \textit{largest} element in the induced partial order, so that for two nodes $\alpha, \beta$ of $T_x$, we have $\alpha > \beta$ iff there is a path $(x_1,...,x_n)$ in $T_x$ with $\alpha = x_1$ and $\beta = x_n$ such that $x_i$ is a parent of $x_{i+1}$ for $1 \leq i \leq n-1$. Thus, for example, in Figure \ref{fig:basic_rooted_tree}, $x > q$ and $z > s,$ etc.

%Note that the definition of depth that we have chosen for a given rooted tree, $T_x$, is really the maximum depth or \textit{max-depth}. We could equally well have chosen to define the depth as the minimum depth or \textit{min-depth}, in other words, as the minimum number of nodes in a path from the designated root node to any leaf. The min-depth of the rooted tree in Figure \ref{fig:basic_rooted_tree} is $2$. We hope to return to a characterization of min-depth once we have completely characterized max-depth, which henceforth we will just call \textit{depth}.
%
%The discussion that follows applies without change to the more general setting of a forest of rooted trees, and the partial orders induced by such a forest, but we constrain our discussion just to individual rooted trees for simplicity.
%
%In our discussions we shall speak interchangeably of the depth of a rooted tree and the length of the maximum path in the induced partial ordering.

The problem of distinguishing the depth of a rooted tree via a first-order formula with a minimum number of quantifiers is similar to the analogous problem for linear orders of different sizes, since a rooted tree has depth $k$ or greater iff it has a leaf node, above which there is linear order of size at least $k-1$. %Hence this problem is a natural one to turn to after completely characterizing linear orders in terms of the minimum number of quantifiers needed to distinguish them.

Our strategy will be to characterize a tree of depth $d$ recursively as a graph containing a vertex $v$ which has a subtree of depth $k$ that includes $v$ and everything below it, and a linear order of length $d-k$ comprising the vertices above $v$, where $k$ is chosen to minimize the total number of quantifiers. We then show that this is the minimum quantifier way to characterize a tree of each given depth.

%In the next section on rooted trees and forests, we will see that distinguishing rooted trees and forests of different depths via first-order sentences shares much in common with distinguishing linear orders of different sizes. In this vein, 
The following result is classic and key to establishing a number of fundamental inexpressibility results in FOL~\cite{Lib12}. It is typically obtained by appeal to Theorem~\ref{thm:ef}.

%\hl{the above reference is a book, right? it would be good to put a chapter/page number, as well as a reference to a primary source.}

%\hl{We should make completely clear in the intro which results we are merely citing and which are our own results. To this end, we should be putting brackets and references in the theorem statements that are not our own (I am doing this below, but the page numbers and a primary source should also be included there). If there is no specific theorem in the ref, then we should be writing something like} "Follows from~\cite{Lib12}, Theorem ???, page ???"

\begin{theorem}[\cite{Lib12}, Theorem 3.6] \label{thm:f}
Let $f(r) = 2^r -1$.
In an $r$-round  \edashf  game played on two linear orders of different sizes, Duplicator wins if and only if the size of the smaller linear order is at least $f(r)$.
\end{theorem}

Analogs of Theorem \ref{thm:f} are proven for M-S games in \cite{Fagin21}. %\hl{N: I don't think we need the following results? previous line should be sufficient?}
The following definition and theorems are from that paper. %and summarize these results.

\begin{definition}[\cite{Fagin21}] \label{def:g}
Define the function $g: \mathbb{N} \rightarrow \mathbb{N}$ such that $g(r)$ is the maximum number $k$ such that there is a formula with $r$ quantifiers that can distinguish linear orders of size $k$ or larger from linear orders of size less than $k$.
\end{definition}
%
%To see that $g$ is well defined, observe that the sentence 
%\begin{equation} \label{eqn:basic}
%	%\exists x_1 \cdot\cdot\cdot \exists x_r \bigwedge\limits_{1 \leq i < j \leq r} x_i < x_j,
%	\exists x_1 \cdot\cdot\cdot \exists x_r \bigwedge\limits_{1 \leq i < r} x_i < x_{i+1},
%\end{equation}
%distinguishes linear orders of size $r$ or larger from linear orders of size less than $r$. Furthermore, 
%there are only finitely many inequivalent sentences with up to $r$ quantifiers that include only the relation symbols $<$ and $=$, some fraction of which distinguish linear orders of some size $k$ or greater from linear orders of size less than $k$. There is therefore a maximum such $k \geq r$, which is then $g(r)$.

\begin{theorem}[\cite{Fagin21}] \label{thm:g}
The function $g$ takes on the following values: 
$g(1) = 1, g(2) = 2, g(3) = 4, g(4) = 10$, and for $r > 4$,
\begin{equation*}
g(r) = \begin{cases} &2g(r-1)~~~~~~\textrm{if $r$ is even,}\\ &2g(r-1) + 1\textrm{ if $r$ is odd.} \end{cases}
\end{equation*}
%In an $r$-round  multi-structural game played on two linear orders of different sizes, Duplicator wins if and only if the size of the smaller linear order is at least $g(r)$.
\end{theorem}

%The following theorem for multi-structural games is the analog of Theorem \ref{thm:f} for E-F games, and describes precisely when Duplicator (alternatively, Spoiler) wins $r$-round multi-structural games on two linear orders of different sizes.

\begin{theorem}[\cite{Fagin21}] \label{thm:g_for_game_play}
In an $r$-round M-S game played on two linear orders of different sizes Duplicator has a winning strategy if and only if the size of the smaller linear order is at least $g(r)$. 
\end{theorem}

For given positive integers $r$ and $k$, we want to know if there exist sentences with $r$ quantifiers that distinguish rooted trees of depth $k$ or larger from rooted trees of depth smaller than $k$.  For $k = r$, one such sentence is

\begin{equation} \label{eqn:basic}
	%\exists x_1 \cdot\cdot\cdot \exists x_r \bigwedge\limits_{1 \leq i < j \leq r} x_i < x_j,
	\exists x_1 \cdot\cdot\cdot \exists x_r \bigwedge\limits_{1 \leq i < r} (x_i < x_{i+1}),
\end{equation}
which distinguishes rooted trees of depth $r$ or larger from rooted trees of depth less than $r$.  Here, if $T_x$ is a rooted tree of depth exactly $r$ then $x_1$ would be a deepest child. Since there are only finitely many inequivalent formulas in up to $r$ variables that include the relations $<$ and $=$ and at most $r$ quantifiers, 
there is some maximum such $k$, which we shall designate by $t(r)$.  With $\mathbb{N} = \{1, 2, ...\}$, we restate this definition of $t$ formally as follows. Note that no meaningful sentence about trees can be constructed with a single quantifier, so the definition begins at $r = 2$.

\begin{definition} \label{def:t}
Define the function $t:\{2,3,...\}\rightarrow \mathbb{N}$ such that $t(r)$ is the maximum number $z$ such that there is a formula with $r$ quantifiers that can distinguish rooted trees of depth $z$ or larger from rooted trees of depth less than $z$.
%Define the function $g:\mathbb{N} \rightarrow \mathbb{N} \cup \{\infty\}$ such that $g(r)$ is the maximum number $k$, or infinity, such that the Spoiler can win multi-structural games on a pair of linear orders when one linear order is of size $k$ or greater, and the other is of size less than $k$, and, moreover, the Duplicator can win multi-structural games on two linear orders when both linear orders are of size $k$ or greater.
\end{definition}

By (\ref{eqn:basic}) above, $t(r) \geq r$ for $r \geq 2$. For an M-S game of $r$ rounds on rooted trees of sizes $t(r)$ or larger on one side, and $t(r) - 1$ or smaller on the other side, by the Equivalence Theorem, Spoiler will have a winning strategy. 

Since linear orders are perfectly good rooted trees, we have the following:

\begin{observation} \label{obs:t_g_bound}
For all $r$ we have $t(r) \leq g(r)$.
\end{observation}

In the subsections that follow on rooted trees we sometimes refer to the deeper tree or family of trees in a given multi-structural game by $B$ (for ``Big'') and the shallower tree/family of trees by $L$ (for ``Little''). Analogously, when considering games on linear orders, $B$ often refers to the bigger linear order(s) and $L$ the littler one(s).  Further, on linear orders of some size, a designation of the form L4, say, refers to the $4$th smallest element of $L$ and analogously B4 to the $4$th smallest element of $B$. If we have to refer to an element in a variable position, say in the $k+1$st position of $L$, we would write $L(k+1)$.%, as per the convention established in \cite{Fagin21}.

\subsection{Establishing $t(2)$ and $t(3)$} \label{sec:first-t-values}

%\begin{lemma} \label{lemma:t_of_1}
%	$t(1) = 1$.
%\end{lemma}
%
%\begin{proof}
%The formula $\exists x (x = x)$ distinguishes rooted trees of depth $1$ and greater from the empty tree.  Furthermore, %Duplicator wins any one round game on rooted trees of length at least $1$.
%\end{proof}

We establish upper bounds of the form $t(r) \leq k$ %, as was done for linear orders in \cite{Fagin21}, 
by finding specific trees of depths $k$ and $k' > k$, and then finding Duplicator-winning strategies for the associated $r$-round multi-structural game on these trees. %There is an important difference between multi-structural games on linear orders and multi-structural games on rooted trees, however. In the case of linear orders there is a unique structure of size $k$. In the case of rooted trees depth does not uniquely determine a tree and in order for Duplicator to establish that depth-$k$ trees \textit{cannot} be distinguished by a first-order sentence with $r$ quantifiers, we must allow him to pick the trees at the beginning of the game. 

\begin{definition}
By $T_x(d)$ we mean a tree rooted at $x$ of depth $d$. Analogously, $T_x(d+)$ means that the tree rooted at $x$ has depth $d$ or greater, and $T_x(<d)$ means that the tree has depth less than $d$.
\end{definition}

%\begin{proof}
%Since in the $r$-round multi-structural game to distinguish rooted trees of depth $k$ from those of depth $k'$, Duplicator has the freedom to choose the depth $k$ and $k'$ rooted trees, he can just pick the rooted trees to be linear orders. The result follows directly.
%\end{proof}

%Before that, however, let us establish the easier lower bounds.

\begin{lemma} \label{lemma:t_of_2}
$t(2) = 2$.
\end{lemma}

\begin{proof}
The formula $\Upsilon_2 = \exists x \exists y(x < y)$ distinguishes $T_x(<2)$ from $T_x(2+)$ and so $t(2) \geq 2$.  The upper bound $t(2) \leq g(2) = 2$ follows from Observation \ref{obs:t_g_bound} and the known value of $g(2)$ given by Theorem \ref{thm:g}. %The lemma follows.
\end{proof}

\begin{lemma} \label{lemma:t_of_3}
$t(3) = 4$.
\end{lemma}

\begin{proof}
The inequality $t(3) \leq g(3) = 4$ follows from Observation \ref{obs:t_g_bound} and the value of $g(3)$ given by Theorem \ref{thm:g}. To establish $t(3) \geq 4$, let us look at the two different expressions that distinguished \textit{linear orders} of size at least $4$ from those of size less than $4$:
\begin{align}
	\Phi_{3,\forall} = \forall x\exists y \exists z(&x < y < z \vee y < z < x) \label{g3_forall} \\%= \forall x\exists y \exists z\phi_{3,\forall} \label{g3_forall} \\
	\Phi_{3,\exists} = \exists x \forall y \exists z(& \label{g3_exists} \\
		&y < x \rightarrow z > x~~~\land \notag \\
		&y > x \rightarrow (z \neq y \land z > x)~~~\land \notag \\
		& y = x \rightarrow z < x) \notag %= \exists x \forall y \exists z \phi_{3,\exists} \notag 
\end{align}
Note that $\Phi_{3,\forall}$, a statement that says that for every $x$ there are either two elements less than $x$ or two elements greater than $x$, fails for the rooted tree of depth $4$ in Figure \ref{fig:tree_4_1},
\begin{figure} [h]
\centerline{\scalebox{0.40}{\includegraphics{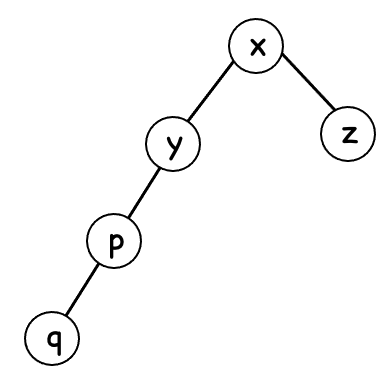}}}
\caption{A rooted tree of depth $4$ that does not satisfy $\Phi_{3,\forall}$.}
\label{fig:tree_4_1}
\end{figure}
since the vertex $z$ fails to satisfy this condition, and hence is not a viable candidate for a formula that distinguishes $T_x(4+)$ from $T_x(<4)$. However, $\Phi_{3,\exists}$ \textit{does} succeed in this regard, since it says that there is an element that has one smaller element and two larger elements. A rooted tree of depth $4$ always has such an element -- the parent of a deepest leaf node -- the element labeled $p$ in the figure. Further, $\Phi_{3,\exists}$ is satisfied by every rooted tree of depth at least $4$ and no rooted tree of depth less than $4$. The lemma follows.
\end{proof}

\subsection{Establishing $t(4)$} \label{app:t_of_4}
% !TEX root = Multi-structural Games ITCS.tex
%The first place where $t(r)$ differs from $g(r)$ is at $r = 4$, as the following lemma shows.

\begin{definition} By l.o.(k) we mean the unique linear order with $k$ nodes.
\end{definition}

\begin{lemma} \label{lemma:t_of_4_upper_bound}
$t(4) \leq 8$.
\end{lemma}

\begin{proof}
We first show that Duplicator can win a $4$-round multi-structural game on the pair of rooted trees of depths $9$ and $10$ depicted in Figure \ref{fig:10_vs_9_trees}. %Let us begin with the base case of $T_x(10)$ \vs  $T_x(9)$. The argument for $T_x(10+)$ \vs $T_x(\leq 9)$ is essentially equivalent. As a first step, Duplicator picks rooted trees of depth $10$ and $9$ as depicted in Figure \ref{fig:10_vs_9_trees}.
\begin{figure} [h]
\centerline{\scalebox{0.40}{\includegraphics{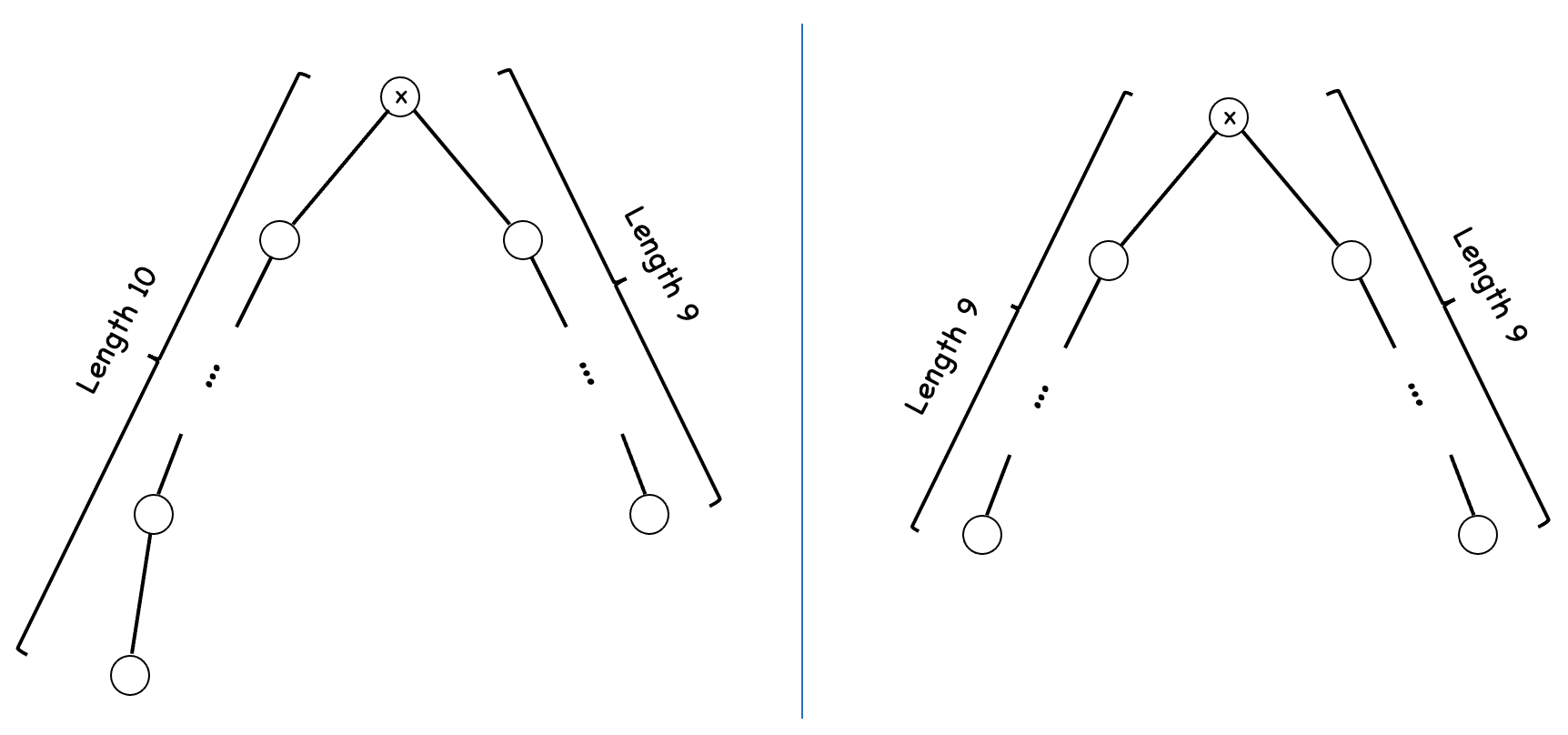}}}
\caption{The case of $T_x(10)$ (left-hand side) \vs $T_x(9)$ (right-hand side). The rooted tree, $T_x(10)$, has two branches, one of length $10$, the other of length $9$. The rooted tree, $T_x(9)$, has two branches each of length $9$.}
\label{fig:10_vs_9_trees}
\end{figure}
If Spoiler plays his first move on $L$, any move choice is mirrored with a symmetrical move on the length-$9$ branch of $B$ and Spoiler's 1st move is essentially wasted. Hence, Spoiler's best move is on the length-$10$ branch of $B$. It is then easy to see that the remainder of the game can be assumed to be played completely on the length-$10$ branch of $B$ and one of the length-$9$ branches of $L$, so that in effect we are playing a $10$ \vs $9$ linear order game where Spoiler plays first on $B$.

Let $B$ now stand for a linear order of size $10$ (i.e., the left branch of $T_x(10)$) and $L$ stand for a linear order of size $9$ (i.e., the left branch of $T_x(9)$). %Within such a game note that if Spoiler starts with an end move on $B$, playing an identical end move on $L$ reduced to a Duplicator-winnable $9$ \vs $8$ game on linear orders. Otherwise, wherever Spoiler plays on $B$, Duplicator will respond by making a second copy of $L$ and then making a move that matches Spoiler's short side on one of the boards and matches Spoiler's long side on the other board. For example, see Figure \ref{fig:10_vs_9_los}.
%For the rest of this game argument we will make Duplicator's job just a bit harder by assuming we are playing a multi-structural game with atoms (see section 6 of \cite{Fagin21} for a description of these games).  
If Spoiler plays on $B$ leaving a short side of $3$ or less then Duplicator can match the short side play on a single copy of $L$ and reach a position with long sides of size $6+$ \vs $5+$ and so 
%by the Reduction Lemma of \cite[Lemma 6.7]{Fagin21} and the fact that $g'(3) = 5$, 
reach a position that is easily seen to be Duplicator-winning by direct play-out. %or allowing Spoiler the slightly greater power of playing ``atoms'' \cite[Section 6]{Fagin21} and observing that Duplicator necessarily wins such an l.o.(6+) \vs l.o.(5+) game. 
Thus, WLOG, we may assume Spoiler plays B5, in which case Duplicator will respond by playing on two copies of $L$, playing L5 on one and L4 on the other, as in Figure \ref{fig:10_vs_9_los},
\begin{figure} [h]
\centerline{\scalebox{0.40}{\includegraphics{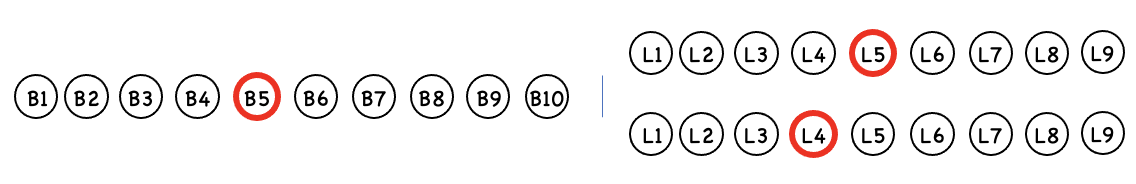}}}
\caption{The $T_x(10)$ \vs $T_x(9)$ game turns into an l.o.(10) \vs l.o.(9) game where Spoiler is constrained to play first on l.o.(10). A most challenging move is to play B5, to which Duplicator responds playing L4 on one copy of $L$ and L5 on another.}
\label{fig:10_vs_9_los}
\end{figure}
in one case matching the short side of the main $B$ branch, and in the other case matching the long side of the main $B$ branch. Although Duplicator can always play with the oblivious strategy, in this case playing just these two moves suffices and simplifies our analysis. If Spoiler now makes his 2nd round play on $B$, a move to the left or on top of B5 is matched with an identical move on the first copy of $L$, while a move to the right of B5 is matched with identical long-side move on a second copy of $L$.  In either case Duplicator easily survives another two rounds just on a single pair of structures. Suppose instead that on his 2nd move, Spoiler plays on $L$. He will clearly want to play on the non-matched sides of each copy of $L$, in other words, playing on L6-L9 on the top copy of $L$ and on L1-L3 on the bottom $L$ copy. (Otherwise he will just transpose into a case considered a moment ago, when Spoiler played his 2nd round move on $B$.) For this analysis Duplicator can ignore the bottom copy of $L$ since she just needs to maintain a single isomorphism. If Spoiler plays L6, Duplicator responds with B6 and clearly survives two more rounds, while a Spoiler move of L7 meets with a Duplicator response of B7, again surviving $2$ more rounds. Spoiler moves of L8 or L9 are met symmetrically with B9 or B10 respectively. Thus Duplicator survives the l.o.($10$) \vs l.o.($9$) game where Spoiler must play first on l.o.($10$) and hence Duplicator also survives the $T_x(10)$ \vs $T_x(9)$ game. Thus $t(4) \leq 9$.

It would be nice, at this point, if we could claim that $t(4) = 9$ by using our expression that established $t(3) \geq 4$ to say that there is an element $w$ with a rooted tree of size at least $4$ both above and below $w$, and in this way differentiate $T_x(9+)$ from $T_x(< 9)$. However, things are not that easy; the expression (\ref{g3_exists}), which established $t(3) \geq 4$, started with an existential quantifier, and the logical expression we would end up using to mimic the aforementioned English language expression would start with two existential quantifiers, and so we wouldn't be able to use it to capture the ``\textit{both} above and below'' part of the English language expression.

%Instead let's approach the problem via a game argument. The Spoiler will choose any path of length $9$ in $B$ and pick the middle element -- in other words an element at depth $5$ -- along this path for his 1st round move. Let's first suppose that Duplicator picks an element at depth less than $5$ from $L$. Then Spoiler can force play entirely above the selected elements so that it becomes a game on one linear order of size 
With the failure of this attempted logical expression in the back of our minds, consider the case of $T_x(9)$ \vs $T_x(8)$ where we pick trees in the same basic model as the $T_x(10)$ \vs $T_x(9)$ trees, but with a bit more nuance. See Figure \ref{fig:19_vs_8_christmas}. 
\begin{figure} [h]
\centerline{\scalebox{0.50}{\includegraphics{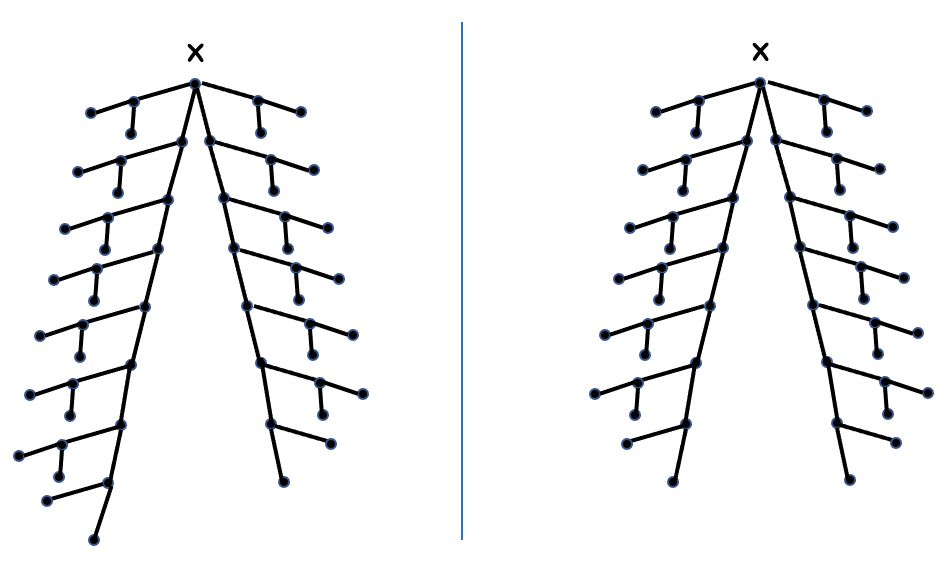}}}
\caption{More nuanced rooted trees in the $T_x(9)$ (left-hand side) \vs $T_x(8)$ (right-hand side) game.}
\label{fig:19_vs_8_christmas}
\end{figure}
The left hand main branch of the left hand tree is of length $9$ while all other main branches of the two trees are of length $8$. As earlier, it is wasteful for Spoiler to play his 1st round move on any of the main branches of length $8$ or their offshoots, and a most challenging move is to select the mid-point along the main $9$ branch in $B$. In essence Spoiler is trying to force the play of a $9$ vs.\ $8$ linear order game in which he is forced to play first on $L$ -- which indeed would be Spoiler-winning. However, as we shall see, the more nuanced trees in Figure \ref{fig:19_vs_8_christmas} provide just enough additional detail so that Duplicator can foil this strategy (because there is now not just a linear order below the selected 1st round nodes, but \textit{rooted trees}). In response, Duplicator will make a second copy of $L$ and play on the 4th element along one of the length $8$ branches in the first copy, call this copy $L_1$, and the 5th element along one of the length $8$ branches in the second copy, which we will refer to as $L_2$.  See Figure \ref{fig:9_vs_8_christmas_split}.
\begin{figure} [h]
\centerline{\scalebox{0.50}{\includegraphics{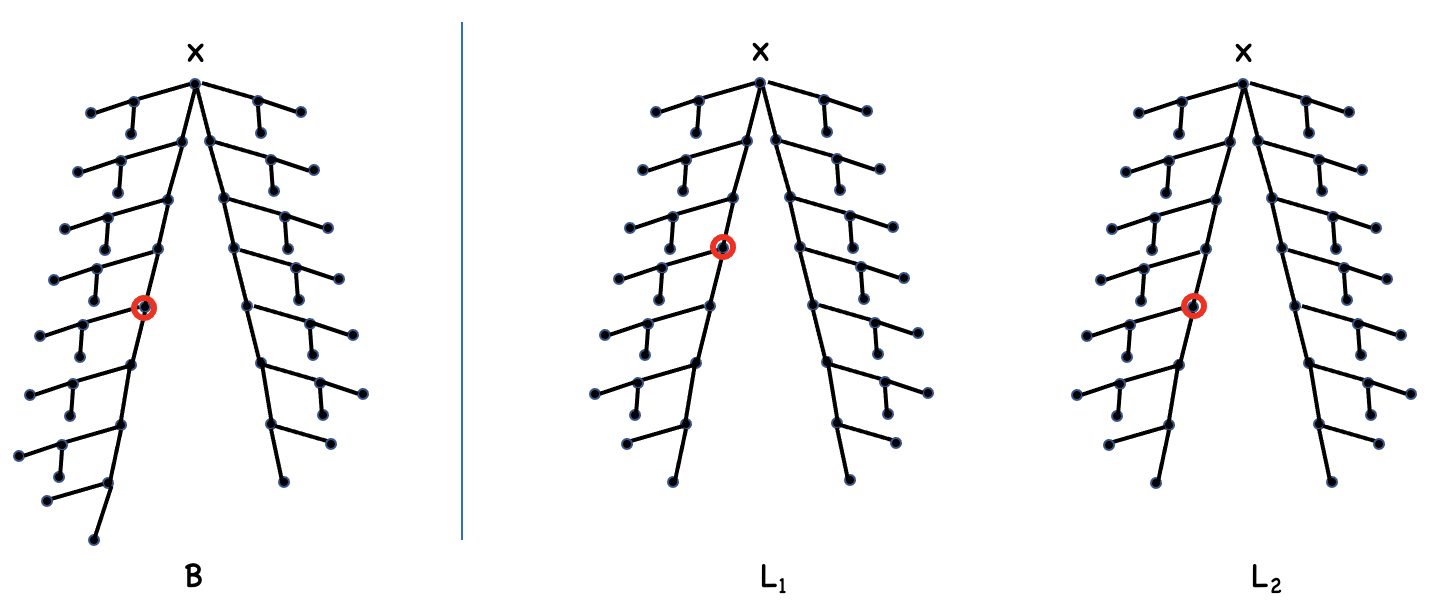}}}
\caption{The $T_x(9)$ \vs $T_x(8)$ game after Spoiler plays the midpoint along the long branch of $B$, while Duplicator makes a second copy of $L$ and plays on the 4th element along one of the length $8$ branches in the first copy, which we call $L_1$, and the 5th element along one of the length $8$ branches in the second copy, which we call $L_2$.}
\label{fig:9_vs_8_christmas_split}
\end{figure}
If Spoiler is to win in the sub-game of $B$ \vs $L_2$ he must be able to win a $3$-round game on the sub-tree below where the first moves were played on these two trees, with the addition of the ability to play on top of a 1st round move, if necessary.  %In fact, we will give Spoiler even more power, namely the ability to play atoms. 
Our aim will be to simply show that Spoiler cannot win in the remaining $3$ rounds in a game of just $B$ \vs $L_2$ by playing first on $L_2$. Suppose otherwise, and note that directing play to the 3-node sub-trees that are depicted to the left of the 1st round-selected nodes is not helpful to Spoiler so we may safely ignore those nodes. The critical sub-trees and Duplicator color-coded responses to the various possible Spoiler 2nd round moves on $L_2$ are given in Figure \ref{fig:4_vs_3_christmas}.
\begin{figure} [h]
\centerline{\scalebox{0.70}{\includegraphics{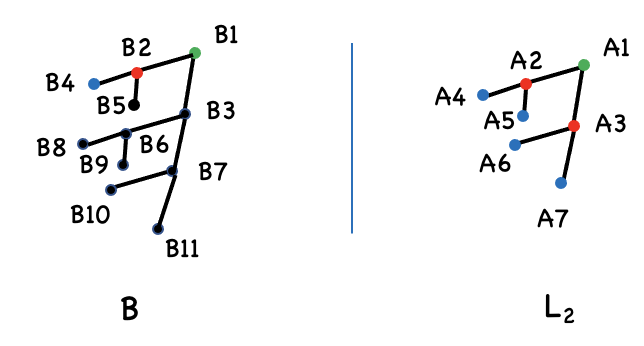}}}
\caption{The critical sub-trees and Duplicator color-coded responses to the various possible Spoiler 2nd round moves on $L_2$}
\label{fig:4_vs_3_christmas}
\end{figure}
(We will consider Spoiler 2nd round moves on $B$ in just a moment.) The selection of a node from $L_2$ by Spoiler is responded to by Duplicator by selecting the node of the same color in $B$. It is easy to see that Duplicator wins in all cases with one minor exception, namely when Spoiler plays, say, A3, Duplicator responds with B2 and now Spoiler plays either either B6 or B7, say B6.  In response to such a move, Duplicator must make a copy of $L_2$ and play A2 on one copy and a move such as A4 on the other copy. The move A2 safeguards against a follow-up of B8 or B9, while A4 safeguards against a follow-up of B3, B7, B10 or B11. It is thus evident that in order for Spoiler to win the $B$ \vs $L_2$ sub-game he must play his 2nd round move on $B$ and select and element somewhere below the element selected in the 1st round. But the only way Spoiler can win the $B$ \vs $L_1$ sub-game is to force the last three moves to be played in the linear orders above the 1st round moves, which now is not possible. Hence Duplicator can win the $T_x(9)$ \vs $T_x(8)$ game and so $t(4) \leq 8$ and the lemma is established.
\end{proof}

\begin{lemma} \label{lemma:t_of_4_lower_bound}
$t(4) \geq 8$.
\end{lemma}

\begin{proof} The following sentence,  with $4$ quantifiers, distinguishes rooted trees of depth $8$ or greater from those of depth less than $8$:
\begin{align} 
T_4 = \exists w \forall x \exists y \exists z(& \label{T_4} \\
	&x > w \rightarrow x < y < z \vee w < y < z < x \label{cond:x_greater_than_w} \\
	&x < w \rightarrow w > y > x \vee x > y > z \label{cond:x_less_than_w} \\
	&x = w \rightarrow y < w \land w < z).\label{cond:x_equal_w}
\end{align}
This sentence says that there exists an element $w$ with a linear order of length $4$ above it, and a rooted tree of depth $3$ below it. The condition attached to the equality condition $x = w$ is also important, and we explain that in a moment too. First, the condition (\ref{cond:x_greater_than_w}) is the analog of equation (\ref{g3_forall}), for $\Phi_{3,\forall}$, described earlier, but relativized to say that there is a linear order of length at least $4$ ``above'' our chosen element $w$. The condition (\ref{cond:x_less_than_w}) says that there is a tree of depth at least $3$ below $w$ by virtue of saying that for every element $x$ below $w$, there is either one element above $x$ and below $w$, or else that there are two additional elements below $x$, one, call it $y$, which is below $w$ and another, call it $z$, which is  below $y$. With just the $x > w$ and $x < w$ implications, we are not guaranteed that there are actually any elements meeting the $x > w$ or $x < w$ conditions. The $x = w$ implication guarantees that there actually are elements meeting both of these conditions. The lemma follows.
\end{proof}

%The following is the immediate consequence Lemmas \ref{lemma:t_of_4_lower_bound} and \ref{lemma:t_of_4_upper_bound}.

\begin{corollary} \label{corollary:t_of_4}
$t(4) = 8$.
\end{corollary}

\subsection{Establishing $t(r)$ -- Generic Case}

In the proof of Theorem \ref{thm:g} \cite{Fagin21}, the authors provide explicit sentences that distinguish linear orders of size $g(r)$ or greater from those of size less than $g(r)$. From the proof of their Theorem 1.6, it can be seen that the distinguishing sentences $\Phi_r$, for $r > 4$ take the form:
\begin{eqnarray*}
	\Phi_r = \begin{cases} \exists x_1 \forall x_2 \cdot\cdot\cdot \forall x_{r-1} \exists x_r \phi_r ~~~~~\textrm{for }r\textrm{ odd,} \\
	\forall x_1 \exists x_2 \cdot\cdot\cdot \forall x_{r-1} \exists x_r \phi_r~~~~~\textrm{for }r\textrm{ even,}\end{cases}
\end{eqnarray*}
where $\phi_r$ is quantifier-free. For odd $r$, the formula $\Phi_r$ says that there exists a point $x_1$, with a linear order of size at least $\lfloor\frac{r}{2}\rfloor$ to both sides of $x_1$. For even $r$, the formula $\Phi_r$ says that for all $x_1$, there exists a linear order of at least size $\frac{r}{2}$ to one side or the other of $x_1$. 

%We define analogous expressions:
%\begin{equation}
%%\Upsilon_1 =& \exists x_1(x_1 = x_1), \\
%\Upsilon_r = \exists x_0 \Phi_{r-1} \land \bigwedge_{i=1}^{r-1} x_0 < x_i, \textrm{for }r  > 1.
%\end{equation}
%$\Upsilon_r$ uses a total of $r$ quantifiers and says that there is an $x_0$ (a deepest leaf) with a linear order of size $g(r-1)$ above it, and therefore distinguishes rooted trees of depth $g(r-1) + 1$ or greater from those of depth $g(r-1)$ or less for $r > 1$. Thus we have the following:
%
%\begin{lemma}
%For $r > 1$,
%\begin{equation*}
%t(r) \geq g(r-1) + 1.% ~~\textrm{for }r > 1.
%\end{equation*}
%\end{lemma}
%
%\noindent This lemma provides the bounds $t(2) \geq 2, t(3) \geq 3, t(4) \geq 5$. By virtue of Lemma \ref{lemma:t_of_3} and Corollary \ref{cor:t_of_4}, the lower bounds for $t(3)$ and $t(4)$ are \textit{not} tight. 

Let us denote by $t_\forall(r)$ the maximum number $k$ such that rooted trees of depth $k$ and above can be distinguished from rooted trees of depth less than $k$ using prenex formulas with $r$ quantifiers beginning with a universal quantifier. Equivalently,  $t_\forall(r) = k$ is the largest depth of a rooted tree such that Spoiler has a winning strategy on  $r$-round M-S games played on rooted trees of depth $k$ or greater versus those of depth less than $k$ when his first move is constrained to be on the tree of lesser depth. Analogously, when considering linear orders, let $g_\forall(r)$ and $g_\exists(r)$ denote, respectively, the maximum number $k$ such that linear orders of size $k$ and above can be distinguished from linear orders of size less than $k$ using prenex formulas with $r$ quantifiers beginning, respectively, with a universal or existential quantifier. 

%Note that formulas that begin with universal quantifiers can never distinguish a structure from an empty structure. The corresponding game would result in no possible starting move for the Spoiler. This caveat is implicit in all that follows.

\begin{lemma} \label{t_forall_recurrence}
For $r > 1$, one has $t_\forall(r) = t(r-1) + 1$.
\end{lemma}

\begin{proof} Given an $r$-round multi-structural game played on rooted trees of depth $k$ and $k+1$, we can choose two rooted trees, identical to the two trees in Figure \ref{fig:19_vs_8_christmas} of Appendix \ref{app:t_of_4}, but with main branches of lengths $k+1$ and $k$ rather than $9$ and $8$. %The tree of depth $k$ will have two identical main branches, each of this same depth. On the other hand, the tree of depth $k+1$ will have a main branch of depth $k+1$ and the other of depth $k$ that is an exact copy of the main branches on the smaller tree. 
Note that a Spoiler 1st move on the smaller tree is completely wasted unless the move chosen is the top node. Choosing any other node on the smaller tree can be exactly mirrored by playing the analogous move on the right hand side of the big tree.  To establish a non-isomorphism, Spoiler must force play to the left hand side of the deeper tree, after which, play on the right hand side would be of no consequence. In response to a top node 1st move, Duplicator will be forced to choose the top node from the larger tree. The problem then reduces to distinguishing a tree of depth $k-1$ (the sub-tree whose top node is just below the top node along the longest branch) vs. trees of depth $k-2$, but where Spoiler may now play anywhere. The lemma follows.
\end{proof}

%\input{t_of5}

%\input{tree_calcs}

%We establish the following dual relations via a simultaneous induction.

\begin{lemma} \label{weak_g_induction}
For $r \geq 2$, the following hold:
\begin{align}
g_\exists(2r) & = 2g_\forall(2r-1) + 1,  \\
g_\forall(2r+1) & = 2g_\exists(2r).
\end{align}
Further, there are expressions establishing the $g_\exists$ relations having prenex signatures $\exists\forall\cdot\cdot\cdot\exists\forall\exists\exists$ with $r-1$ iterations of the $\exists\forall$ pair and then a final $\exists\exists$, while there are expressions establishing the $g_\forall$ relations having prenex signatures $\forall\exists\cdot\cdot\cdot\forall\exists\exists$ with $r$ iterations of the $\forall\exists$ pair and then a final $\exists$.
\end{lemma}

\begin{proof}
We prove these relations by simultaneous induction, starting with the base case of $g_\exists(4) = 2g_\forall(3) + 1$. The inequality $g_\forall(3) \geq 4$ follows from (\ref{g3_forall})% in Appendix \ref{sec:first-t-values}
, while $g_\forall(3) \leq g(3) = 4$ so that $g_\forall(3) = 4$, and we must therefore establish that $g_\exists(4) = 9$. In what follows, we imagine our linear orders stretching from left (the smallest element) to right (the largest element). After an element is selected in a given round, we refer to the two remaining ``sides'' after playing that element as the elements that are either all smaller than (and hence to the left of) the selected element or all greater than (and hence to the right of) the selected element.  We certainly can write an expression stating that ``there exists an $x$ with a linear order of size $g_\forall(3)$ to either side'' so that $g_\exists(4) \geq 9$. Equality is established by showing that Duplicator wins the $4$-round l.o.(10) \vs l.o.(9) game with Spoiler constrained to play first on $B$ -- but this is precisely what is shown in the second paragraph of the proof of Lemma \ref{lemma:t_of_4_upper_bound}.% in Appendix \ref{app:t_of_4}.
%
%To establish equality we show that Duplicator wins the $4$-round $10$ vs. $9$ game with Spoiler constrained to play first on $B$. %We have examined this situation already in Figure \ref{fig:10_vs_9_forall}. 
%Refer to Figure \ref{fig:10_vs_9_los} from Appendix \ref{app:t_of_4}.
%Spoiler's most testing move is B5 (equivalently B6), to which Duplicator replies with the two moves L5 and L4.  Spoiler then is constrained to play next on $L$ and to try to win the l.o.(5) vs. l.o.(4) $3$-round game, which is easily seen to be impossible. Thus $g_\exists(4) = 9$ is established.

The general (inductive) $g_\exists(2r)$ argument is essentially the same argument as we have just given. Instead of arriving at an l.o.(5) vs. l.o.(4) $3$-round game we arrive at an $(r-1)$-round l.o.($g_\forall(2r-1) + 1$) vs. l.o.($g_\forall(2r-1)$) game where Spoiler must play first on $L$, which is Duplicator-winnable by the induction hypothesis.

For the $g_\forall(2r+1)$ argument we consider linear orders of sizes $2g_\exists(2r)$ and $2g_\exists(2r)+1$ where Spoiler must play first on $L$. Spoiler's best move is to play $L(g_\exists(2r))$ (or $L(g_\exists(2r)+1)$) since Duplicator will always respond by matching the shorter side of any play thereby forcing further play to the longer side and hence Spoiler is best off keeping the two sides as balanced as possible. Duplicator will then respond by playing $B(g_\exists(2r))$ on one copy of $B$ and $B(g_\exists(2r)+1)$ on a second copy. Spoiler then must play next on $B$ and try to win a $2r$-round l.o.($g_\exists(2r) + 1$) vs l.o.($g_\exists(2r)$) game. But this game is Duplicator-winnable by the induction hypothesis.

To establish the result about the prenex signatures, observe that the argument got bootstrapped from the expression (\ref{g3_forall}) for $g_\forall(3)$, which has prenex signature $\forall\exists\exists$. Putting together $g_\exists(2r)$ from the two copies of $g_\forall(2r-1)$ tacks an $\exists$ on the front: we have an expression $\exists x_1(\forall x_2 \exists x_3 \exists x_4 \phi \land \forall x_2 \exists x_3 \exists x_4 \phi')$ where $\phi$ and $\phi'$ are the analogs of (\ref{g3_forall}), saying that there is a linear order above and below $x_1$. We pull the sequence of quantifiers $\forall x_2 \exists x_3 \exists x_4$ out in front as follows:
\begin{eqnarray}
    \exists x_1\forall x_2 \exists x_3 \exists x_4( x_2 < x_1 &\rightarrow& x_2 < x_3 < x_4 < x_1 \vee x_3 < x_4 < x_2 < x_1~~\land \label{c4.1}\\ 
    x_2 > x_1 &\rightarrow& x_1 < x_2 < x_3 < x_4 \vee x_1 < x_3 < x_4 < x_2~~\land \label{c4.2}\\
    x_2 = x_1 &\rightarrow& x_3 > x_1 \land x_4 < x_1). \label{c4.3}
\end{eqnarray}
Condition (\ref{c4.1}) says that, assuming there is an element smaller than $x_1$, then there is a linear order of size $3$ smaller than $x_1$. Condition (\ref{c4.2}) says that, assuming there is an element larger than $x_1$, then there is a linear order of size $3$ larger than $x_1$. The equality condition (\ref{c4.3}) guarantees that there are elements both greater than and less than $x_1$. In an analogous fashion one may put together $g_\forall(2r+1)$ from two copies of $g_\exists(2r)$ by tacking on a $\forall$ in front. The lemma follows.
\end{proof}

%In particular, $g_\forall(5) = 18, g_\forall(7) =74, g_\forall(9) = 298$.

%\begin{lemma} \label{prenex-sginature-lemma}
%For $k \geq 1$, there are expressions establishing each of the following function values that have the given prenex signatures:\\
%
%\smallskip
%
%\begin{tabular}{|l|r|}
%  \hline 
% Function Value & Prenex Signature\\
% \hline 
%$g_\forall(2k+1)$ & $\forall\exists \cdot\cdot\cdot\forall\exists\exists$ \\
%$g_\forall(2k)$ & $\forall\exists \cdot\cdot\cdot\forall\exists$ \\
%$t_\forall(2k+1)$ & $\forall\exists \cdot\cdot\cdot\forall\exists\exists$ \\
%$t_\forall(2k)$ & $\forall\exists \cdot\cdot\cdot\forall\exists$ \\
%  \hline 
%\end{tabular}
%\end{lemma}
%
%\begin{proof}
%We have established the signature of $g_\forall(2k+1)$ in Lemma \ref{weak_g_induction}. The LICS paper showed that the signature asserted in the lemma for $g_\forall(2k)$ is the signature for $g(2k)$ for $k \geq 2$, thereby establishing that $g_\forall(2k) = g(2k)$ and hence the signature for  $g_\forall(2k)$ for $k \geq 2$. For $k = 1, g_\forall(2) = g(2)$ is established via the expression $\forall x\exists y(x < y \vee y < x)$.
%
%\end{proof}

\begin{theorem}\label{thm:t_rec}
For $r \geq 2$, the following holds:
\begin{equation} \label{t_1st}
t(r) = g_\forall(r-1) + t_\forall(r-1) + 1 = g_\forall(r-1) + t(r-2) + 2.
\end{equation}
If $r$ is odd, then $g_\forall(r-1) = g(r-1)$ so for $r$ odd we have:
\begin{equation} \label{t_2nd}
t(r) = g(r-1) + t_\forall(r-1) + 1 = g(r-1) + t(r-2) + 2.
\end{equation}

%\begin{eqnarray}
%t(2k) & = & g_\forall(2k-1) + t_\forall(2k-1) + 1 =  g_\forall(2k-1) + t(2k-2) + 2, \\
%t(2k+1) & = & g_\forall(2k) + t_\forall(2k) + 1 =  g(2k) + t(2k-1) + 2.
%\end{eqnarray}
\end{theorem}

\begin{proof}
The first-order sentence establishing the lower bound associated with (\ref{t_1st}), in other words where the left hand equality symbol is replaced by $\geq$, says that ``there exists an element $x$ with a linear order of size $g_\forall(r-1)$ above it, and a rooted tree of depth $t_\forall(r-1)$ below it.'' Lemma \ref{weak_g_induction} established the prenex signature $\forall\exists\cdot\cdot\cdot\forall\exists\exists$ for the expressions $g_\forall(r-1)$ in case $r-1$ is odd. If $r-1$ is even, the Fagin et al. paper \cite{Fagin21} established the prenex signature  $\forall\exists\cdot\cdot\cdot\forall\exists$ for $g(r-1)$, starting with $r-1 = 4$. It follows that $g_\forall(r-1) = g(r-1)$ for such values and so the formula establishing $g_\forall(r-1)$ has this same prenex signature for even values of $r-1$. In case $r-1 = 2$, the value $g(2)$ can be established via the sentence $\forall x \exists y(x < y \vee y < x)$, so that here again $g(2) = g_\forall(2)$ is established via a sentence of the same prenex signature. 

On the other hand, $t_\forall(3) \geq 3$ is established via the sentence $\forall x \exists y \exists z(x < y \vee y < z < x)$ with prenex signature $\forall\exists\exists$, while $t(3) \geq 4$ is established via the expression (\ref{g3_exists}), with prenex signature $\exists\forall\exists$, and hence, by the proof of Lemma \ref{t_forall_recurrence}, the lower bound for $t_\forall(4)$ is established via the prenex signature $\forall\exists\forall\exists$. The first-order sentence, described in English at the beginning of this proof, provides a means for turning an expression for $t_\forall(r)$ of a given prenex signature into an expression for $t(r+1)$ with the same prenex signature but with a leading $\exists$ added. By the proof of Lemma \ref{t_forall_recurrence}, $t_\forall(r+2)$ is then obtained by tacking another $\forall$ in front. Hence the expressions for $t_\forall(r)$ maintain consistent prenex signatures based on their parity, and $t_\forall(r)$ and $g_\forall(r)$ will inductively have identical prenex signatures as long as we can, simultaneously, inductively establish the theorem.

We thus define the expressions $\Phi_{r-1,\forall}$ for $g_\forall(r-1)$, and $\Upsilon_{r-1,\forall}$ for $t_\forall(r-1)$, inductively with the same prenex signatures, e.g.,
\begin{eqnarray}
\Phi_{r-1,\forall} & = & \forall x_{r-1} \exists x_{r-2} \cdot\cdot\cdot \exists x_1 \phi_{r-1,\forall} \\
\Upsilon_{r-1,\forall} & = & \forall x_{r-1} \exists x_{r-2} \cdot\cdot\cdot \exists x_1 \tau_{r-1,\forall}, 
\end{eqnarray}
where $\phi_{r-1,\forall}$ and $\tau_{r-1,\forall}$ are quantifier-free. In order to form the expression $\Upsilon_r$ for $t(r)$, we must relativize both  $\phi_{r-1,\forall}$ and $\tau_{r-1,\forall}$ so that for the new $x_r$ (``$x$'' in the English language sentence at the beginning of the proof),  $\phi_{r-1,\forall}$ applies for values of $x_{r-1}$ that are greater than $x_r$, while $\tau_{r-1,\forall}$ applies for values of $x_{r-1}$ that are less than $x_r$. Moreover, in the relativized expression for $\phi_{r-1,\forall}$ all variables $x_1,...,x_{r-2}$ must be constrained to be greater than $x_r$, while in the relativized expression for $\tau_{r-1,\forall}$ all variables $x_1,...,x_{r-2}$ must be constrained to be \textit{less} than $x_r$. Let us refer to these relativized versions of $\phi_{r-1,\forall}$ and $\tau_{r-1,\forall}$ as $\phi^{rel}_{r-1,\forall}$ and $\tau^{rel}_{r-1,\forall}$ respectively. With these relative expressions we are able to pull out all of the quantifiers and obtain the expression $\Upsilon_r$ for $t(r)$ as follows:
\begin{align*}
\Upsilon_r = \exists x_r\forall x_{r-1} \exists x_{r-2} \cdot\cdot\cdot \exists x_1(&x_{r-1} > x_r \rightarrow \phi^{rel}_{r-1,\forall}~~\land \\
													& x_{r-1} < x_r \rightarrow \tau^{rel}_{r-1,\forall}).
\end{align*}
This expression establishes that for $r \geq 2$, we have $t(r) \geq g_\forall(r-1) + t_\forall(r-1) + 1$.  In order to establish that $t(r) \leq g_\forall(r-1) + t_\forall(r-1) + 1$ we  show that Duplicator can win multi-structural games on rooted trees of depths $t(r)$ and $t(r+1)$ that are the analogs of the trees in Figure \ref{fig:19_vs_8_christmas}. To have a chance of winning an $r$-round game on such symmetric trees, Spoiler must force play to the longest branch of $B$, in other words, force play to the branch of $B$ of length  $g_\forall(r-1) + t_\forall(r-1) + 2$. %and a branch of $L$ that is of depth $g_\forall(k-1) + t_\forall(k-1) + 1$. 
The only Spoiler move on $L$ that would force such an outcome would be to select the very top node.  If this were an optimal play, then we would have $t(r) \leq t_\forall(r-1) + 1 \leq g_\forall(r-1) + t_\forall(r-1) + 1$. The only other way for Spoiler to force play onto the longest branch of $B$ is for him to play his 1st move directly on $B$. If Spoiler were then to leave a linear order of size at least $g_\forall(r-1) + 1$ above the played move, then Duplicator can make a copy of $L$ and on one copy play a move that leaves the identical tree below the played move to the tree left on $B$ and a linear order of size $g_\forall(r-1)$ above, and on the other copy leaves a liner order of size $g_\forall(r-1) + 1$ above and a tree of depth one less than that left on $B$. Spoiler would then be forced to play next on $L$ and would have to win a $g_\forall(r-1)$ \vs $g_\forall(r-1) + 1$ $(r-1)$-round game playing first on $L$, which is impossible by the definition of $g_\forall(r-1)$. On the other hand, if Spoiler leaves a tree of depth at least $t_\forall(r-1) + 1$ below, then Duplicator wins down there by the parallel argument incorporating the definition of $t_\forall(r-1)$. On a branch of length at least $g_\forall(r-1) + t_\forall(r-1) + 2$, leaving a linear order above of length at least $g_\forall(r-1) + 1$ or a tree of depth $t_\forall(r-1) + 1$ below is unavoidable. The upper bound on $t(r)$ is thus established and so, for $r \geq 2$, we have $t(r) = g_\forall(r-1) + t_\forall(r-1) + 1$.

The fact that $t(r) = g_\forall(r-1) + t_\forall(r-1) + 1 =  g_\forall(r-1) + t_\forall(r-2) + 2$ follows by Lemma \ref{t_forall_recurrence}. If $r$ is odd then $r-1$ is even and, as we have remarked earlier in this proof, then $g_\forall(r-1) = g(r-1)$. The theorem follows.
\end{proof}
\medskip

%We can combine the previous results to prove the following explicit expression for $t()$. The proof can be found in Appendix~\ref{subsec:t_explicit}.

\begin{theorem}\label{thm:t_explicit}
For all $r \geq 1$ we have 
%\begin{align*}
%    t(2r) = \frac{7\cdot4^r}{18} + \frac{4r}{3} - \frac{8}{9}, \\
%    t(2r+1) = \frac{8\cdot4^r}{9} + \frac{4r}{3} - \frac{8}{9}.
%\end{align*}
\begin{align*}
    t(2r) = \frac{7\cdot4^r}{18} + \frac{4r}{3} - \frac{8}{9},~~~~
    t(2r+1) = \frac{8\cdot4^r}{9} + \frac{4r}{3} - \frac{8}{9}.
\end{align*}
\end{theorem}

\begin{proof}
As established in Lemmas \ref{lemma:t_of_2}, \ref{lemma:t_of_3} and Corollary \ref{corollary:t_of_4} of Sections \ref{sec:first-t-values} and \ref{app:t_of_4}, we have $t(2)=2, t(3)=4$ and $t(4)=8$.

Let us consider the even case in the statement of the theorem first. By Theorem~\ref{thm:t_rec} we have the for all $r \geq 3$:
\begin{align*}
    t(2r) &= g_{\forall}(2r-1)+t(2r-2)+2\\
    &= t(4)+\sum\limits_{i=2}^{r-1} (g_{\forall}(2i+1)+2)
\end{align*}

It follows from Lemma \ref{weak_g_induction} that $g_\forall(3) = 4$ for $r \geq 2$,
\begin{eqnarray}
g_\forall(2r+1) & = & 2g_\exists(2r) \notag \\
& = &2(2g_\forall(2r-1)+1) \notag \\
& = & 4g_\forall(2r-1) + 2.
\end{eqnarray}

Solving this linear recurrence yields $g_\forall(2r+1) = \frac{7\cdot4^r}{6}-\frac{2}{3}$. Plugging this in and simplifying gives us that for all $r \geq 1$,
$$t(2r) = \frac{7\cdot4^r}{18} + \frac{4r}{3} - \frac{8}{9}.$$
An analogous argument for the odd case gives us, again for all $r \geq 1$,
$$t(2r+1) = \frac{8\cdot4^r}{9} + \frac{4r}{3} - \frac{8}{9}.$$
\end{proof}

\subsection{Comparison of the Growth Rates of the Functions $f,g$ and $t$} \label{subsec:fgt-table}

The following table compares the values of the functions $f,g$ and $t$ for $2 \leq r \leq 10$. Recall that $f(r)$ %(Theorem \ref{thm:f}) 
is the maximum value such that an expression of quantifier rank $r$ can distinguish linear orders of size $f(r)$ or greater from linear orders of size less than $f(r)$, while $g(r)$ 
%(Theorem \ref{thm:g}) 
is the maximum value such that an expression with $r$ quantifiers can distinguish linear orders of size $g(r)$ or greater from linear orders of size less than $g(r)$.

\begin{footnotesize}
\smallskip
\begin{tabular}{r|r|r|r}
 $r$&$f(r)$&$g(r)$&$t(r)$ \\
  \hline 
%  1 & 1 & 1 & 1  \\
  2 & 3 & 2 & 2 \\
  3 & 7 & 4 & 4\\
  4 & 15 & 10 & 8\\
  5 & 31 & 21 & 16\\
  6 & 63 & 42 & 28\\
  7 & 127 & 85 & 60\\
  8 & 255 & 170 & 104 \\
  9 & 511 & 341 & 232 \\
  10 & 1023 & 682 & 404
\end{tabular}
\end{footnotesize}
\smallskip

\section{s-t Connectivity} \label{sec:st-conn}

In this section we explore the number of quantifiers needed to express either directed or undirected $s$-$t$ connectivity (henceforth STCON) in FOL with the binary edge relation $E$, as a function of the number $n$ of edges in a shortest path between the distinguished nodes $s$ and $t$. %Since we will include $s$ and $t$ in this count, $n \geq 2$. 
STCON, also known as \textit{reachability} between labelled nodes $s$ and $t$, refers to the property of graphs that labelled nodes $s$ and $t$ are connected. STCON$(n)$ denotes the property that $s$ and $t$ are connected by a path of length at most $n$ edges.

In Appendix \ref{app:log_2_method}, we show how to describe STCON$(n)$ using $2\log_2(n) + O(1)$ quantifiers. %The sentence we give is a slight variation on the typical sentence used (\cite{Lib12}) to provide a $\log_2(n) + O(1)$ upper bound on the minimum quantifier rank needed to express STCON$(n)$. 
The following theorem generalizes that construction, to improve the number of quantifiers to $3\log_3(n) + O(1)$. A similar argument shows that $K\log_K(n) + O(1)$ quantifiers can be used for any positive integer $K$, although this quantity is minimized for $K=3$.

\begin{theorem} \label{thm:st-con-upper} STCON(n) can be expressed with $3 \log_3(n)+O(1)$ quantifiers.
\end{theorem}

% !TEX root = Multi-structural Games ITCS.tex

\begin{proof}
%In this section %, since we will be considering expressions with different quantifier structures, 
We shall use $\Upsilon_i$ to denote sentences with $i$ quantifiers, and $\tau_j$ to denote quantifier-free expressions, where the subscript $j$ denotes the path length characterized by $\tau_j$.

We start with the following simple expression stating that $s$ and $t$ are connected and $d(s,t) \leq 3$, where $d(s,t)$ denotes the length of the shortest path from $s$ to $t$: 
\begin{equation}
\Upsilon_2 = \exists x_1 \exists x_2(\tau_3 \vee \tau_2 \vee \tau_1),
\end{equation}
where:
\begin{eqnarray}
	\tau_3 &=& E(s,x_1) \land E(x_1, x_2) \land E(x_2, t), \label{tau3} \\
	\tau_2 &=& E(s,x_1) \land E(x_1, t), \label{tau2} \\
	\tau_1 &=& E(s,t).  \label{tau1}
\end{eqnarray}
%The construction of sentences is now very similar to the construction in the prior section. 
We now iteratively add three quantifiers at each stage and slot two nodes between each of the previously established nodes, as in Figure \ref{fig:s-to-x1-to-x2-to-t}. 
\begin{figure} [ht]
\centerline{\scalebox{0.50}{\includegraphics{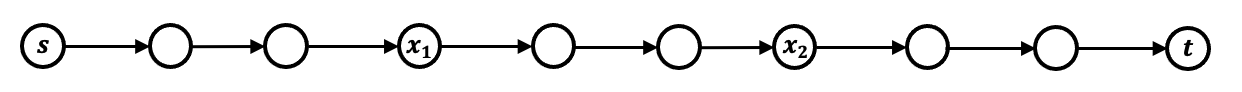}}}
\caption{An illustration of slotting two nodes between each of the pre-established nodes $s, x_1, x_2$ and $t$ in order to express a distance $9$, $s-t$ path, using $5$ quantifiers as in expressions (\ref{Tau5}) and (\ref{tau9}).}
\label{fig:s-to-x1-to-x2-to-t}
\end{figure}
%This can be done in much the same manner as for the case of $2\log_2(n)$ quantifiers (Appendix \ref{app:log_2_method}) but by utilizing one additional antecedent condition in the collection of implications.
%To illustrate, we can express that there is a path of distance at most $9$ from $s$ to $t$, using only $5$ quantifiers:
We express that there is a path of length at most $9$ from $s$ to $t$, using $5$ quantifiers as follows:
\begin{equation} \label{Tau5}
\Upsilon_5 = \exists x_1 \exists x_2 \forall x_3 \exists x_4 \exists x_5 (\tau_9 \vee \tau_8 \vee \cdot\cdot\cdot \vee \tau_1).
\end{equation}
In this case, we just show $\tau_9$. The simplifications required to get from $\tau_8$ down to $\tau_4$ are analogous to those for getting from $\tau_3$ down to $\tau_1$, but where we apply  (\ref{tau3}) -- (\ref{tau1}) separately to each of   (\ref{tau9}) -- (\ref{ineq_cond}).
\begin{align}
\tau_9 = \quad\quad & ((x_3 = s) \rightarrow E(s,x_4) \land E(x_4, x_5) \land E(x_5, x_1))~~\land \label{tau9}\\
				& ((x_3 = t) \rightarrow E(x_1,x_4) \land E(x_4, x_5) \land E(x_5, x_2))~~\land \label{tau9.1}\\
				& ((x_3 \neq s \land x_3 \neq t) \rightarrow E(x_2,x_4) \land E(x_4, x_5) \land E(x_5, t)).  \label{ineq_cond}
\end{align}

Using $8$ quantifiers, we can slot two new nodes between each node established in the prior step, as depicted in Figure \ref{fig:s-to-t-by-3-iter2}. 
\begin{figure} [ht]
\centerline{\scalebox{0.50}{\includegraphics{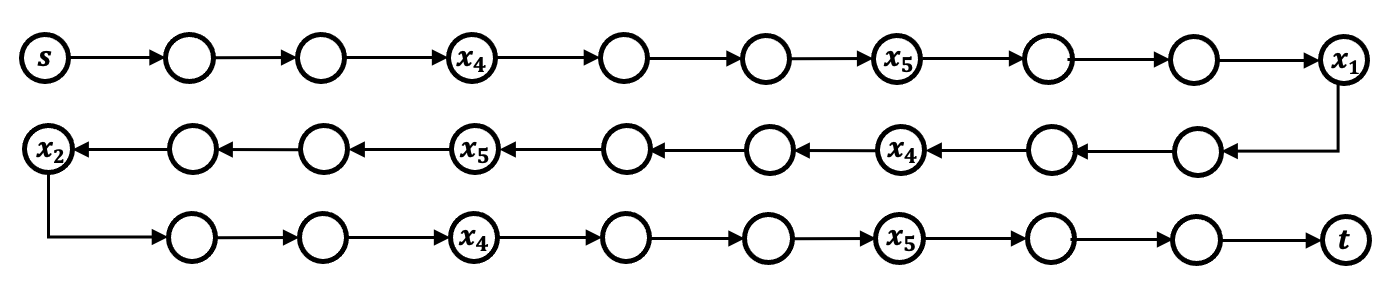}}}
\caption{Slotting two nodes between each of the pre-established nodes $s, x_1, x_2, x_4, x_5$, and $t$ in order to express a distance $27$, $s-t$ path, using $8$ quantifiers as in expressions (\ref{Tau8}) and (\ref{tau27})-(\ref{final_clause}).}
\label{fig:s-to-t-by-3-iter2}
\end{figure}
The associated logical expression is
\begin{equation} \label{Tau8}
\Upsilon_8 = \exists x_1 \exists x_2 \forall x_3 \exists x_4 \exists x_5 \forall x_6 \exists x_7 \exists x_8 (\tau_{27} \vee \tau_{26} \vee \cdot\cdot\cdot \vee \tau_1),~\textrm{and}
\end{equation}
\begin{align}
\tau_{27} = \quad\quad & ((x_3 = s \land x_6 = s) \rightarrow E(s,x_7) \land E(x_7, x_8) \land E(x_8, x_4))~~\land \label{tau27}\\
				& ((x_3 = s \land x_6 = t) \rightarrow E(x_4,x_7) \land E(x_7, x_8) \land E(x_8, x_5))~~\land \\
				& ((x_3 = s \land (x_6 \neq s \land x_6 \neq t)) \rightarrow E(x_5,x_7) \land E(x_7, x_8) \land E(x_8, x_1))~~\land \\ 
				&~~~~~...  \\
				&(((x_3 \neq s \land x_3 \neq t) \land (x_6 \neq s \land x_6 \neq t)) \rightarrow E(x_5,x_7) \land E(x_7, x_8) \land E(x_8, t)). \label{final_clause}
\end{align}
The expression $\tau_{27}$ will have the two ``pivot points'' around the universally quantified variables $x_3$ and $x_6$ and so have $3^2 = 9$ antecedent conditions corresponding to the possible ways the universally quantified variables $x_3$ and $x_6$ can each take the values $s, t$ or neither $s$ nor $t$. The right hand side of each equality condition describes how to fill in the edges in Figure \ref{fig:s-to-t-by-3-iter2} with two new vertices (utilizing the two newest existentially quantified variables, $x_7$ and $x_8$) and three new edges.

In this way we obtain sentences with $3n - 1$ quantifiers that can express STCON instances of path length up to $3^n$.
%$3n+2$ quantifiers that can express STCON instances of path length $3^{n+1}$. 
%Taking base-$3$ logarithms, we see that 
Thus, when $n$ is a power of $3$, we can express STCON instances of length $n$ with $3\log_3(n) -1 \approx 1.893\log_2(n) - 1$ quantifiers, and when $n$ is \textit{not} a power of $3$, with $\lfloor 3\log_3(n) + 2 \rfloor$ quantifiers. The theorem therefore follows. 
\end{proof}

\paragraph*{A Remark on Lower Bounds on Quantifier Rank and hence on Number of Quantifiers}

Lower bounds on the number of quantifiers for $s$-$t$ connectivity follow readily from the literature. The well-known proof that connectivity is not expressible in FOL (\cite[Prop. 6.15]{Imm99} or \cite[Corollary 3.19]{Lib12}) can be used to establish that $s$-$t$ connectivity with path length $n$ is not expressible as a formula of quantifier rank $\log_{2}(n) - c$ for some constant $c$. %We include a proof for completeness in Appendix~\ref{app:lb-proof}.

\begin{theorem}[Immerman, Proposition 6.15~\cite{Imm99}] \label{thm:st-lb}
There exists a constant $c$ such that $s$-$t$ connectivity to path length $n$ is not expressible as a formula of quantifier rank $\log_{2}(n) - c$.
\end{theorem}

Since the quantifier rank is a lower bound on the number of quantifiers, the previous theorem immediately implies a lower bound on the number of quantifiers as well.
%\noindent This result directly establishes the same lower bound for number of quantifiers. 
While we have shown that STCON($n$) can be expressed with $3\log_3(n) + O(1)$ quantifiers, we note that the minimum quantifier \emph{rank} of STCON($n$) is well-known to be lower.
\begin{theorem}[\cite{Lib12}] \label{thm:st-ubqr}
$s$-$t$ connectivity to path length $n$ can be expressed with a formula of quantifier rank $\log_2(n)+O(1)$.
\end{theorem}
%A proof of the above theorem,  which is an adaptation of the $2\log_2(n) + O(1)$ construction of Appendix \ref{app:log_2_method}, is given in Appendix \ref{app:st-con-and-q-rank}. A similar proof is given in \cite{Lib12}.
%\subsection{A few Tidbits on $s$-$t$ Connectivity, L and NL}

\section{Final Comments and Future Directions}\label{sec:fina- comments}

Although progress on M-S games did not come until 40 years after their initial discovery in \cite{Immerman81}, the results of this paper show that these games are quite amenable to analysis, and the more detailed information they give about the requisite quantifier structure has the potential to yield many new and interesting insights.

Theorem~\ref{thm:discrepency} tells us that the number of quantifiers can be more than exponentially larger than the quantifier rank. This shows that the number of quantifiers is a more refined measure than the quantifier rank, and gives an interesting and natural measure of the complexity of a FO formula.  It would be interesting to find explicit examples where the quantifier rank is $k$, but where the required number of quantifiers grows even faster than in our example in the proof of Theorem~\ref{thm:discrepency}. %, say where the required number of quantifiers is double exponential in $k$ (e.g., $2^{2^{p(k)}}$ for some polynomial $p(k)$), or triple exponential in $k$ (e.g., $2^{2^{2^{p(k)}}}$ for some polynomial $p(k)$).  
Ideally, we would even like to find explicit examples where the required number of quantifiers is non-elementary in $k$. 

%We have taken a modest step in extending 
We have extended
the results on the number of quantifiers needed to distinguish linear orders of different sizes \cite{Fagin21} to distinguish rooted trees of different depths. Can this line of attack be carried further to incorporate other structures, say to other structures with induced partial orderings such as finite lattices?

The most immediate question arising from our work is whether one can improve the known upper or lower bounds on the number of quantifiers needed to express $s$-$t$ connectivity. In particular, what is the smallest constant $c \geq 1$ such that $s$-$t$ connectivity (up to path length $n$) is expressible using $c \log_2(n)$ quantifiers? Our Theorem \ref{thm:st-con-upper} shows that $Quants(\text{STCON($n$)})$ is at most $3\log_3(n) + O(1) \approx 1.893 \log_2(n) + O(1)$. The well-known lower bound of $Rank(\text{STCON($n$)}) \geq \log_2(n) - O(1)$ (cited as Theorem \ref{thm:st-lb}) 
yields the only lower bound we know on $Quants(\text{STCON($n$)})$, but we also know the upper bound $Rank(\text{STCON($n$)}) \leq \log_2(n)+O(1)$ (cited as Theorem~\ref{thm:st-ubqr}). 
As these upper and lower bounds for the quantifier rank of STCON$(n)$ essentially match, in order to improve the lower bound on $Quants(\text{STCON($n$)})$ further (by a multiplicative constant), we \emph{cannot} rely on a rank lower bound: we will have to resort to other methods, such as M-S games. %to give a better lower bound.

%In our $3\log_3(n)$ quantifier upper bound on $s$-$t$ connectivity (Section~\ref{sec:3log3}), the resulting first-order formulas have size $\Omega(n)$, whereas the $3 \log_2(n)$-quantifier formulas obtained from the old Stockmeyer-Meyer construction (as mentioned in Section~\ref{sec:st-conn}) have size only $O(\log n)$. Is there an inherent tradeoff between formula size and number of quantifiers for expressing $s$-$t$ connectivity? 
%, which directly implies the same lower bound for number of quantifiers. 
%Since these upper and lower bounds for quantifier rank essentially match, in order to improve the quantifier number lower bound further (by a multiplicative constant), we cannot rely on a quantifier rank lower bound: we will have to resort to other methods, such as multi-structural games. %to give a better lower bound.

Another question is whether we can find other problems with even larger quantifier number lower bounds than logarithmic ones. %(In the case of $s$-$t$ connectivity, we cannot hope to prove a quantifier lower bound  greater than $3.1 \log_3(n)$.) %, we could hope to prove larger lower bounds for other problems. 
Let us stress that substantially larger lower bounds on the number of quantifiers would have major implications for circuit complexity lower bounds. For example, by the standard way of expressing uniform circuit complexity classes in FOL~\cite{Imm99}, a property (over the $<$ relation) that requires $\log^{\alpha(n)}(n)$ quantifiers, where $\alpha(n)$ is an unbounded function of $n$, would imply a lower bound for $\text{uniform}_{\FO}\text{-}\mathsf{NC}$. See Appendix~\ref{app:eq} for an exact statement.

Another interesting direction to push this research is to extend the notion of multi-structural games to 2nd-order logic, first-order logic with counting or to fixed point logic. %, and see if anything interesting emerges.

%A graph is $2$-connected if you can remove any node and the graph remains connected.
%
%\begin{pquestion} If we add a predicate indicating that two nodes of a graph are connected, does this lead us to new interesting questions? In particular, can we then express $2$-connectivity in first-order logic?
%\end{pquestion}
%
%It is unclear how to express properties of the form ``removing a vertex gives....'' or ``adding a vertex with such and such connectivity properties gives....'' in first-order, or even 2nd-order logic. Note though that being 2-connected is in P so in NP and thus by Fagin's Theorem, expressible in $\exists$SO.  How?

%Ron asked a similar (though less ambitious) question:

%\begin{question} Is there a family of graphs on $n$ vertices that requires $n$ quantifiers to characterize?
%\end{question}

%\begin{question} If we extend the new games to FOL + Counting or to fixed point logics do we get anything interesting?
%\end{question}

%\pagebreak
\bibliography{ref}

\section*{Appendix}

\appendix

\section{The Number of Sentences in Vocabulary $V$ with $k$ Quantifiers is at Most Doubly Exponential in $k$} \label{app:double_exp_upper_bd}

The double exponential bound is obtained as follows. If a first-order sentence has $k$ quantifiers, then it can be written as $Q_1 x_1 \ldots Q_k x_k \phi$, where each $Q_i$ is a quantifier (either $\forall$ or $\exists$), and where $\phi$ is a quantifier-free formula in so-called \textit{full disjunctive normal form}, in other words, a disjunction of conjunctions, where each possible atomic formula in $V$ appears, either negated or not negated, in each of the conjunctions. The number of possible initial quantifier sequences $Q_1 x_1 \ldots Q_k x_k$ is $2^k$. If the vocabulary $V$ has $c$ relation symbols, each of arity at most $r$, then the number $A$ of atomic formulas is at most $c k^r$. So the number $C$ of conjunctions, which each contain either the positive or negated form of each atomic formula, is $2^A$.  A disjunction of these conjunctions corresponds to a selection of a subset of them and there are therefore at most $2^C$ of these. Hence, in  total, the number of sentences with $k$ quantifiers is at most $2^k 2^C = 2^k 2^{2^A} = 2^{(k+{2^{ck^r}})}< 2^{2^ {k + ck^r}}$.
This gives us a double exponential upper bound.  
Although it is not needed for our purposes, we note that a slight modification of this argument gives a double exponential lower bound, even when all of the quantifiers are existential.

\section{Expressing STCON(n) Using $2\log_{2}(n) + O(1)$ Quantifiers} \label{app:log_2_method}
% !TEX root = Multi-structural Games ITCS.tex
%Firstly, let us establish that we can express $s$-$t$ connectivity (henceforth STCON)
%using $2\log_2(n)$ quantifiers. We will do this through a sequence of expressions that we will inductively explain how to form. 
We start by showing that we can describe STCON in the case where the number of edges, $n = 2$, using a single quantifier:
\begin{equation} \label{Phi1}
\Phi_1 = \exists x ((E(s,x) \land E(x,t)) \vee E(s,t)).
\end{equation}
Here, and in subsequent expressions, the index of $\Phi_i$ refers to the number of quantifiers in the expression.  To understand the more complicated cases it is useful to write expression (\ref{Phi1}) in the following form:
\begin{equation} \label{Phi1.1}
\Phi_1 = \exists x (\phi_2 \vee \phi_1),
\end{equation}
where $\phi_2 = E(s,x) \land E(x,t)$ is the distance-$2$ part of the unquantified expression, and  $\phi_1 =  E(s,t)$ is the distance-$1$ part of the unquantified expression.

Analogously, for the case $n = 4$, we have
\begin{equation} \label{Phi3}
\Phi_3 = \exists x \forall y \exists z(\phi_4 \vee \phi_3 \vee \phi_2 \vee \phi_1),
\end{equation}
where $\phi_1$ and $\phi_2$ are as in (\ref{Phi1.1}), and
\begin{align}
\phi_4 =\quad\quad & y=s \rightarrow E(s,z) \land E(z,x)~~~\land \label{phi4} \\
		&y=t \rightarrow E(x,z) \land E(z,t),	\notag \\
\phi_3 =\quad\quad & y=s \rightarrow E(s,z) \land E(z,x)~~~\land  \label{phi3}\\
		&y=t \rightarrow E(x,t). \notag			
\end{align}

To understand what the above sentence is saying, consider for the moment the sentence (\ref{Phi3}), but without the disjuncts for $\phi_3, \phi_2$ and $\phi_1$, which express that $s$ and $t$ are connected at the respective distances $3,2$ and $1$. In the sentence $\exists x \forall y \exists z \phi_4$, remember that we have $5$ nodes. In this sentence, we are declaring the existence of an element $x$ that is the central node in a path from $s$ to $t$, as depicted in Figure \ref{fig:s-to-x-to-t}.
\begin{figure} [h]
\centerline{\scalebox{0.60}{\includegraphics{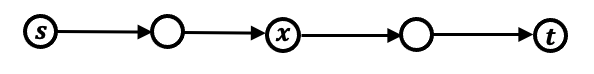}}}
\caption{The five node path from $s$ to $t$, with middle node $x$.}
\label{fig:s-to-x-to-t}
\end{figure}
Now $x$ is fixed, but depending on what $y$ is, $z$ can play different roles. Thus, if $y = s$ we use $z$ to guarantee a ``bridge'' from $s$ to $x$, in the sense of there being edges $E(s,z), E(z,x)$, and in case $y=t$, we use $z$ to guarantee a ``bridge'' from $x$ to $t$ in the analogous sense that there are edges $E(x,z), E(z,t)$. While the existentially quantified variable $x$ has a fixed interpretation, the universally quantified variable $y$ allows us to ``pivot'' in either of two directions and in so doing, the existentially quantified variable $z$ can play exactly two roles.

Returning to the expression (\ref{phi3}), $\phi_3$ drops an arbitrary one of the aforementioned ``bridges,'' and is true if and only if $d(s,t) = 3$. $\phi_2$ and $\phi_1$ were defined to support (\ref{Phi1.1}) and remain unchanged.
%, saying, respectively, that $d(s,t) = 2$ and $d(s,t) = 1$. 
Thus, in (\ref{Phi3}), $\Phi_3$ says that $d(s,t)$ is either $1,2,3$ or $4$ -- hence that $d(s,t) \leq 4$, as claimed.

As we introduce, inductively, successive quantifier alternations, the universal quantifier will serve to provide more ``pivot points'' to enable exponentially more of these length-$2$ bridges, as we shall see. In turn, the existentially quantified variable following each universal quantifier will be committed to the midpoint associated with each ``gap.'' To see how this plays out in the case of $\Phi_5$, which expresses STCON when $d(s,t) \leq 8$, we have
\begin{equation} \label{Phi5}
\Phi_5 = \exists x_1 \forall x_2 \exists x_3 \forall x_4 \exists x_5(\phi_8 \vee \phi_7 \vee \phi_6 \vee \phi_5 \vee \phi_4 \vee \phi_3 \vee \phi_2 \vee \phi_1)
\end{equation}
where we have replaced the earlier variables $x,y$ and $z$ by $x_1, x_2$ and $x_3$.  In a picture, the analog of the prior Figure \ref{fig:s-to-x-to-t} is Figure \ref{fig:s-to-x3-to-x1-to-x3-to-t}. 
\begin{figure} [h]
\centerline{\scalebox{0.60}{\includegraphics{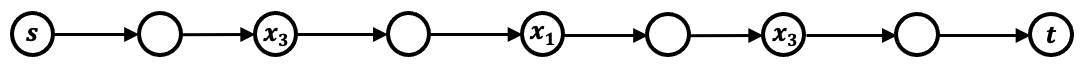}}}
\caption{The $9$ node, distance $8$, path from $s$ to $t$. The locations of $x_1$ and $x_3$ are ``committed'' as a result of the first three quantifiers, leaving four ``gaps'' that are filled as a result of the four possible ``pivots'' associated with the second universal quantifier -- i.e., the universal quantifier quantifying over $x_4$.}
\label{fig:s-to-x3-to-x1-to-x3-to-t}
\end{figure}
The first universal quantifier enables a first pivot, as we saw in the expression $\exists\forall\exists\phi_4$, allowing two possible placements of $z$ (now labelled $x_3$), while the second universal quantifier enables a second pivot, which allows for four possible locations for $x_5$. The full expression is as follows:

\begin{align}
\phi_8 = \quad\quad& (x_2 = s \land x_4 = s) \rightarrow E(s,x_5) \land E(x_5, x_3)~~\land \label{c1} \\
			       & (x_2 = s \land x_4 = t) \rightarrow E(x_3,x_5) \land E(x_5, x_1)~~\land \label{c2} \\
			       & (x_2 = t \land x_4 = s) \rightarrow E(x_1,x_5) \land E(x_5, x_3)~~\land \label{c3} \\
			       & (x_2 = t \land x_4 = t) \rightarrow E(x_3,x_5) \land E(x_5, t) \label{c4}.
\end{align}
Condition (\ref{c1}) establishes the ``bridge'' from $s$ to $x_3$, condition (\ref{c2}) establishes the ``bridge'' from $x_3$ to $x_1$, and so on. Now, for $\phi_7$, we replace the right hand side of (\ref{c4}) with $E(x_3,t)$;  for $\phi_6$, in addition to the replacement (\ref{c4}), we will replace the right hand side of (\ref{c3}) with $E(x_1, x_3)$, and analogously for $\phi_5$, where we will additionally replace the right hand side of  (\ref{c2}) with $E(x_3, x_1)$. The expressions for $\phi_1$ through $\phi_4$ remain as previously described for $\Phi_1$ and $\Phi_3$ (but with the change of variables $x \mapsto x_1, y \mapsto x_2, z \mapsto x_3$.

Now, suppose we have defined $\Phi_{2n+1} = \exists x_1 \forall x_2 \cdot\cdot\cdot \forall x_{2n} \exists x_{2n+1}(\phi_{2^{n+1}} \vee \cdot\cdot\cdot \phi_1)$. Then, as we consider $\Phi_{2n+3}$, the analog of Figure \ref{fig:s-to-x3-to-x1-to-x3-to-t} is Figure \ref{fig:s-to-long-to-t},
\begin{figure} [h]
\centerline{\scalebox{0.50}{\includegraphics{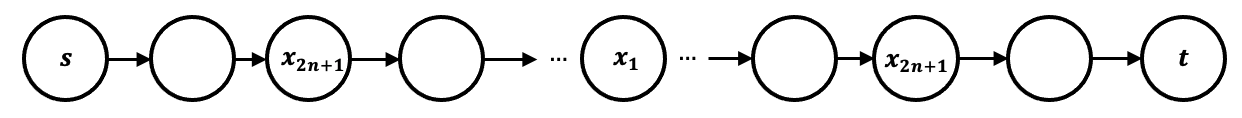}}}
\caption{The $2^{n+2} + 1$ node, distance $2^{n+2}$, path from $s$ to $t$. The locations associated with $x_1, x_3,...,x_{2n+1}$ are ``committed'' as a result of the first $2n+1$ quantifiers, leaving $2^{n+1}$ ``gaps'' that are filled as a result of the $2^{n+1}$ possible ``pivots'' associated with the final universal quantifier -- i.e., the universal quantifier quantifying over $x_{2n+2}$, which in turn determine the possible locations for $x_{2n+3}$ -- the variable associated with the final existential quantifier.}
\label{fig:s-to-long-to-t}
\end{figure}
and we can define
\begin{equation} \label{Phi2n+3}
\Phi_{2n+3} = \exists x_1 \forall x_2 \cdot\cdot\cdot \forall x_{2n+2} \exists x_{2n+3}(\phi_{2^{n+2}} \vee \cdot\cdot\cdot \phi_1),
\end{equation}
where
\begin{align}
\phi_{2^{n+2}} = \quad& (x_2 = s \land x_4 = s \land \cdot\cdot\cdot \land x_{2n+2} = s) \rightarrow E(s,x_{2n+3}) \land E(x_{2n+3}, x_{2n+1})~~\land  \label{Phi2n+3.1}\\
			       &(x_2 = s \land x_4 = s \land \cdot\cdot\cdot \land x_{2n+2} = t) \rightarrow E(x_{2n+1},x_{2n+3}) \land E(x_{2n+3}, x_{2n-1})~~\land   \label{Phi2n+3.2}\\
			       & \cdot\cdot\cdot \notag \\
			      &(x_2 = t \land x_4 = t \land \cdot\cdot\cdot \land x_{2n+2} = t) \rightarrow E(x_{2n+1},x_{2n+3}) \land E(x_{2n+3}, t). \label{Phi2n+3.n}
\end{align}

\medskip

Thus, for $n \geq 0$, with $2n+1$ quantifiers we can describe an STCON instance of distance $2^{n+1}$. Hence, taking logs to the base $2$, we see that we can express STCON on a graph with $n$ vertices using $2\log_2(n) + O(1)$ quantifiers, for a very small constant $O(1)$.

%\subsection{Expressing $s$-$t$ Connectivity to Path Length $n$ Using $3\log_{3}(n) + O(1)$ Quantifiers - Proof of Theorem %\ref{thm:st-con-upper}} \label{subsec:3log3}
%\input{3log_3(n)}

%\subsection{Proof of Theorem \ref{thm:st-ubqr}: $s$-$t$ connectivity to path length $n$ can be expressed with a formula of quantifier rank $\log_2(n)+O(1)$.} \label{app:st-con-and-q-rank}
%\input{st-con-and-q-rank}

\section{Equivalence of First-Order Logic and Uniform Circuit Complexity Classes}\label{app:eq}
\textbf{Notation:}  By $\FO[\{R_i\}]$ we mean the set of first-order logic sentences of constant size (independent of the size of the structure) using relations $\{R_i\}$. By $\text{uniform}_{\FO}\text{-}\mathcal{C}$ we refer to the first-order uniform version of a complexity class $\mathcal{C}$. (Informally, the connection language of circuits in this class is definable by a first-order sentence. See Definition 5.17 in~\cite{Imm99} for an exact definition.) We will always assume that our formulas have equality as a logical relation, hence whenever we have the $<$ relation we will also have the $\leq$ relation. A single binary string will be defined by a unary relation $\mathbf{1}(x)$ over the domain $[n] := \{1,\ldots,n\}$  which is true if and only if the $x^\textrm{th}$ position of the binary string is a $1$.  Each such relation corresponds directly to a unique $n$-bit string. In our first-order formulas over binary strings, we also allow the following relations over the domain $[n]$, with fixed interpretations:
\begin{enumerate}[(i)]
    \item $x < y$: binary relation which is true if and only if position $x$ occurs strictly before position $y$.
    \item $\text{BIT}(x, y)$: binary relation which is true if and only if the $y^\textrm{th}$ bit of $x$ is 1. 
\end{enumerate}
%When we consider the $<$ and the $\mathbf{1}$ relation we will refer to it as 
We say that $\FO[<]$ is the set of first-order formulas over binary strings (represented by the relation $\mathbf{1}$) with the $<$ relation. $\FO[<, \text{BIT}]$ is the set of first-order formulas over binary strings (represented by $\mathbf{1}$) with both the $<$ and BIT relation.

\begin{definition}[(Definition 4.24~\cite{Imm99})]\label{def:fotn}
Let $b : {\mathbb N} \rightarrow {\mathbb N}$. $\FO[\{R_i\}][b(n)]$ is the set of first-order logic formulas $\Phi$ (with the set of relations $\{R_i\}$) of the form
$$\Phi(x_1, x_2, \ldots, c_k) = (\Psi_1(x_1,\ldots,x_k))^{} \Psi_2(x_1, x_2, \ldots, x_k)$$
where $\Psi_1$, $\Psi_2$ are \emph{strings of logical symbols}, $|\Psi_1|, |\Psi_2|, k$ are all $O(1)$ (independent of the size of the structure), and $(\Psi_1(x_1,\ldots,x_n))^{}\Psi_2(x_1, x_2, \ldots, x_k)$ denotes the concatenation of $\Psi_1$ for $b(n)$ times followed by $\Psi_2$. \end{definition}

An important point in the above definition is that $\Psi_1, \Psi_2$ are \emph{strings of logical symbols} (e.g. $\exists$, $\forall$, $x$, $\vee$, $\wedge$, etc.): they need not be well-formed formulas themselves, but the concatenation $\Phi$ (as described above) must be well-formed.

To make more sense of this definition, we give a specific example. In the following example, we use $.$ after a quantifier in the following manner: By $(\exists x . \psi_1(x)) \psi_2(x)$ we mean $\exists x (\psi_1(x) \rightarrow \psi_2(x))$ and by $(\forall x . \psi_1(x)) \psi_2(x)$ we mean $\forall x (\psi_1(x) \rightarrow \psi_2(x))$.

\textbf{Example: }We can express $s$-$t$ connectivity in a graph of size $n$ as 
$$R(s, t) = (\Psi_1(s, t))^{\ceil{\log_2(n)+1}} \Psi_2(s, t)$$
where 
$$ \Psi_1(s, t) = (\exists z) (\forall a, b . ((a=s \land b=z) \vee (a=z \land b=t))) (\exists s, t. (s=a \land t=b))$$ and 
$$\Psi_2(s, t) = E(s, t) \vee (s = t).$$
%Note that $\FO = \FO[O(1)]$. 

Note that every formula in $\FO[\{R_i\}][b(n)]$ can be written as a formula with at most $O(b(n))$ quantifiers. 

From Barrington, Immerman and Straubing~\cite{BarringtonIS88} it is known that over the structure of binary strings, $ \text{uniform}_{\FO}\text{-}\AC^0 = \FO[<, \text{BIT}]$ where $\text{uniform}_{\FO}\text{-}\AC^0$ refers to first-order uniform $\AC^0$. In fact, for all polynomially bounded and first-order time constructible functions $t(n)$ (Theorem 5.22~\cite{Imm99}),
$\text{uniform}_{\FO}\text{-}\AC[t(n)]$ equals $\FO[<, \text{BIT}][t(n)]$,
where $\text{uniform}_{\FO}\text{-}\AC[t(n)]$ refers to first-order uniform $\AC^0$ circuits of depth $O(t(n))$.

\begin{lemma}[Theorem 5.22~\cite{Imm99}] \label{lem:fo-ac}
$\text{uniform}_{\FO}\text{-}\AC[t(n)] = \FO[<, \text{BIT}][t(n)]$.
\end{lemma}

Note that the above equivalence has the BIT operator, which we did not use in the rest of the paper. The rest of the section shows how to express $\text{uniform}_{\FO}\text{-}\NC$ functions (uniform circuits with $\plog(n)$ depth) without the BIT relation.

\begin{theorem}
Every circuit in $\text{uniform}_{\FO}\text{-}\NC$ over binary strings has an equivalent formula with $\plog(n)$ quantifiers using only the $<$ and $\mathbf{1}$ relations.
\end{theorem}

\begin{proof}
By Lemma~\ref{lem:fo-ac}, every circuit in $\text{uniform}_{\FO}\text{-}\NC$ has an equivalent $\FO[<, \text{BIT}][\plog(n)]$ formula.
By Definition~\ref{def:fotn}, every formula $\Phi(x_1, x_2,\ldots,x_k) \in \FO[<, \text{BIT}][\plog(n)]$ can be expressed as 
$$(\Psi_1(x_1, x_2, \ldots, x_k))^{\plog(n)} \Psi_2(x_1, x_2, \ldots, x_k)$$
where $\Psi_1$, $\Psi_2$ are strings, and $|\Psi_1|, |\Psi_2|, k$ are all bounded by constants. Thus the BIT relation occurs at most $\plog(n)$ times in the entire formula $\Phi$. It is well-known that $\text{BIT} \in \FO[<][\log n]$ (Exercise 4.18 in~\cite{Imm99}): that is, BIT can be expressed with a first-order formula of $O(\log n)$ size. Replacing each occurrence of BIT with the equivalent formula from $\FO[<][\log n]$ yields a formula with at most $O(\plog(n) \cdot \log(n) = \plog(n))$ quantifiers. Hence every circuit in $\text{uniform}_{\FO}\text{-}\NC$ has an equivalent formula with $\plog(n)$ quantifiers over only the $<$ relation.
\end{proof}

\end{document}